\documentclass{article}

\usepackage{microtype}
\usepackage{graphicx}
\usepackage{subfigure}
\usepackage{booktabs} 

\usepackage{hyperref}

\usepackage[accepted]{icml2019}

\usepackage{amsmath}
\usepackage{amsthm}
\usepackage{amssymb}

\def\ind{\mathbb{I}}
\def\P{\mathbb{P}}
\def\E{\mathbb{E}}
\def\eqdef{\overset{\text{def}}{=}}

\newtheorem{theorem}{Theorem}
\newtheorem{lemma}{Lemma}
\newtheorem{remark}{Remark}

\newtheorem{proposition}{Proposition}

\icmltitlerunning{Adaptive Monte Carlo Multiple Testing}

\begin{document}

\twocolumn[
\icmltitle{Adaptive Monte Carlo Multiple Testing via Multi-Armed Bandits}



\icmlsetsymbol{equal}{*}

\begin{icmlauthorlist}
\icmlauthor{Martin J. Zhang}{ee}
\icmlauthor{James Zou}{ee,bmi,biohub}
\icmlauthor{David Tse}{ee}
\end{icmlauthorlist}

\icmlaffiliation{ee}{Department of Electrical Engineering, Stanford University}
\icmlaffiliation{bmi}{Department of Biomedical Data Science, Stanford University}
\icmlaffiliation{biohub}{Chan-Zuckerberg Biohub}

\icmlcorrespondingauthor{James Zou}{jamesyzou@gmail.com}
\icmlcorrespondingauthor{David Tse}{dntse@stanford.edu}

\icmlkeywords{Multiple hypothesis testing, Monte Carlo, permutation test, multi-armed bandits, GWAS}

\vskip 0.3in
]



\printAffiliationsAndNotice{}  

\begin{abstract}
Monte Carlo (MC) permutation test is considered the gold standard for statistical hypothesis testing, especially when standard parametric assumptions are not clear or likely to fail. However, in modern data science settings where a large number of hypothesis tests need to be performed simultaneously, it is rarely used due to its prohibitive computational cost. In genome-wide association studies, for example, the number of hypothesis tests $m$ is around $10^6$ while the number of MC samples $n$ for each test could be greater than $10^8$, totaling more than $nm$=$10^{14}$ samples. In this paper, we propose  \texttt{A}daptive \texttt{M}C multiple \texttt{T}esting (\texttt{AMT}) to estimate MC p-values and control false discovery rate in multiple testing. The algorithm outputs the same result as the standard full MC approach with high probability while requiring only $\tilde{O}(\sqrt{n}m)$ samples. This sample complexity is shown to be optimal. On a Parkinson GWAS dataset, the algorithm reduces the running time from 2 months for full MC to an hour. The \texttt{AMT} algorithm is derived based on the theory of multi-armed bandits.
\end{abstract}

\section{Introduction}
Monte Carlo (MC) permutation testing is considered the gold standard for statistical hypothesis testing. It has the broad advantage of estimating significance non-parametrically, thereby safeguarding against inflated false positives \cite{dwass1957modified,davison1997bootstrap,boos2000monte,lehmann2006testing,phipson2010permutation}. It is especially useful in cases where the distributional assumption of the data is not apparent or likely to be violated. 

A good example is genome-wide association study (GWAS), whose goal is to identify associations between the genotypes (single nucleotide polymorphisms or SNPs) and the phenotypes (traits) \cite{visscher201710}. For testing the association between a SNP and the phenotype, the p-value is often derived via closed-form methods like the analysis of variance (ANOVA) or the Pearson's Chi-squared test \cite{purcell2007plink}. However, these methods rely on certain assumptions on the null distribution, the violation of which can lead to a large number of false positives \cite{yang2014pboost,che2014adaptive}. MC permutation test does not require distributional assumption and is preferable in such cases from a statistical consideration \cite{gao2010avoiding}. However, the main challenge of applying MC permutation test to GWAS is {\em computational}.

MC permutation test is a special type of MC test where the p-values are estimated by MC sampling from the null distribution --- permutation test computes such MC samples by evaluating the test statistic on the data points but with the responses (labels) randomly permuted. Let $T^{\text{obs}}$ be the observed test statistic and $T^{\text{null}}_{1}, T^{\text{null}}_{2}, \cdots, T^{\text{null}}_{n}$ be $n$ independently and identically distributed (i.i.d.) test statistics randomly generated under the null hypothesis. The MC p-value is written as 
\begin{align}\label{eq:fMC_p_val}
    P^{\text{MC}(n)} \eqdef \frac{1}{n+1} \left( 1 + \sum_{j=1}^n \ind {\{T^{\text{null}}_{j} \geq T^{\text{obs}} \}}\right),
\end{align}
which conservatively estimates the ideal p-value $P^{\infty} \eqdef \P(T^{\text{null}} \geq T^{\text{obs}})$. In addition, $P^{\text{MC}(n)}$ converges to the ideal p-value $P^{\infty}$ as the number of MC samples $n \rightarrow \infty$.


GWAS is an example of large-scale multiple testing: each SNP is tested for association with the phenotype, and there are many SNPs to test. For performing $m$ such tests simultaneously, the data is collected and each of the $m$ null hypotheses is associated with an ideal p-value (Fig.\ref{fig:schema}a). A common practice, as visualized in Fig.\ref{fig:schema}b, is to first compute an MC p-value for each test using $n$ MC samples and then apply a multiple testing procedure to the set of MC p-values  $\{P_i^{\text{MC}(n)}\}$ to control the false positives, e.g., using the Bonferroni procedure \cite{dunn1961multiple} or the Benjamini-Hochberg procedure (BH) \cite{benjamini1995controlling}. Here, as folklore, the number of MC samples $n$ is usually chosen to be at least 10 or 100 times\footnote{For a hypothesis with ideal p-value $P^{\infty}$, the relative error for the MC p-value with $n$ MC samples is $1 / \sqrt{n P^{\infty}}$ and therefore, choosing e.g. $n=100/P^{\infty}$ gives a relative error of 0.1. Since in multiple testing the p-values we are interested in can be as small as $1/m$, it is recommended to set $n = 100 m$.} of the number of tests $m$. In GWAS, there are around $10^6$ SNPs to be examined simultaneously via multiple testing and $n$ is recommended to be at least $10^8$ \cite{johnson2010accounting}. The total number of MC samples is $nm$=$10^{14}$, infeasible to compute.

%


\begin{figure}
    \centering
    \centerline{\includegraphics[width=0.48\textwidth]{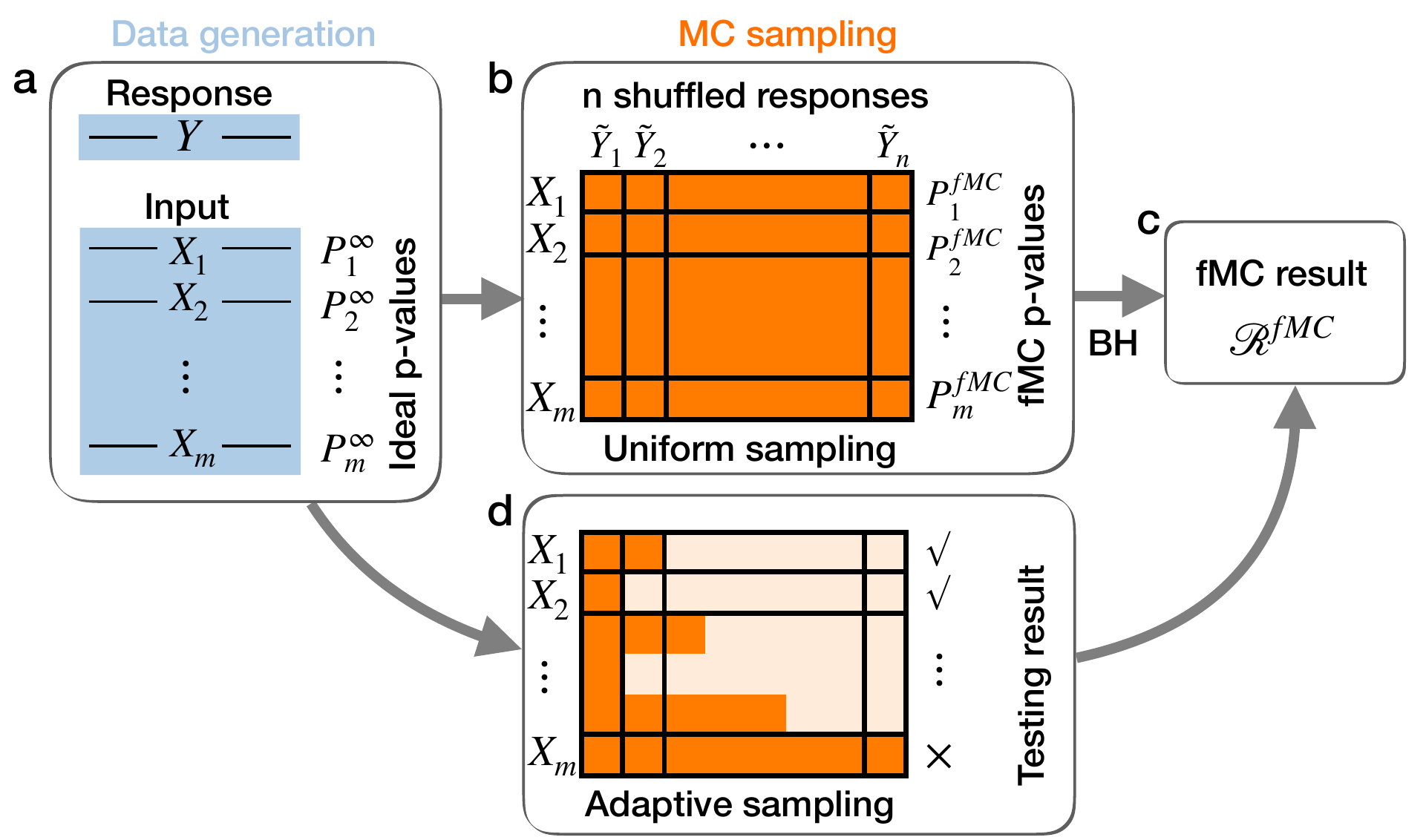}}
    \caption{Workflow. 
    (a) $N$ data samples are collected and the $k$th data point has a response (label) $Y^{(k)}$ and $m$ inputs (features) $(X_1^{(k)}, \cdots, X_m^{(k)})$. 
    There are $m$ null hypotheses to test; the $i$th null hypothesis corresponds to no association between the $i$th input $\mathbf{X}_i$=$(X_i^{(1)}, \cdots, X_i^{(N)})$ and the response $\mathbf{Y}$=$(Y^{(1)}, \cdots, Y^{(N)})$.
    (b) The standard fMC workflow is to: 1) compute an MC p-value for each test $i$ using $n$ MC samples $T^{\text{null}}_{i1}, \cdots, T^{\text{null}}_{in}$; 2) apply the BH procedure on the set of MC p-values to control FDR. In this example, $T^{\text{null}}_{ij}$ is the correlation between the $i$th input $\mathbf{X}_i$ and a randomly permuted response $\mathbf{Y}_{\sigma_j}$. 
    (c) The fMC result is to make discovery (claim association) for a subset of inputs.
    (d) \texttt{AMT} directly estimates the fMC p-values by adaptive MC sampling and recovers the fMC testing result with high probability.}
    \label{fig:schema}
\end{figure}

This work considers the standard full MC (fMC) workflow, as shown in Fig.\ref{fig:schema}b, of first computing p-values with MC sampling and then controlling the false discovery rate (FDR) by applying the BH procedure to the set of MC p-values. The aim is to reduce the number of MC samples while obtaining the same fMC testing result. The focus of the present paper is on solving a \emph{computational} problem, i.e., accelerating the standard fMC workflow, rather than a \emph{statistical} problem, e.g., improving the power of the test. An alternative goal may be to recover the BH discoveries on the ideal p-values $\{P_i^{\infty}\}$, which is an ill-posed problem that may take unrealistically many MC samples. Recovering the fMC result, however, takes at most $nm$ samples and any improvement over the complexity $nm$ of uniform sampling represents an improvement over the standard workflow. 

{\bf Contribution.} We propose \texttt{A}daptive \texttt{M}C multiple \texttt{T}esting (\texttt{AMT}) to compute the fMC testing result via adaptive MC sampling. While recovering the fMC result with high probability, it effectively improves the sample complexity from $nm$ to $\tilde{O}(\sqrt{n}m)$ under mild assumptions that encompass virtually all practical cases, where $\tilde{O}$ hides logarithmic factors. A matching lower bound is provided. In a GWAS dataset on the Parkinson's disease, it improves the computational efficiency by 2-3 orders of magnitude, reducing the running time from 2 months to an hour. We note that $\texttt{AMT}$ is not specific to MC permutation test; it can be used for MC tests in general.

The fMC procedure computes $n$ MC samples for each of the $m$ null hypotheses. For each null hypothesis, a randomly selected subset of the MC samples can provide an estimate of its fMC p-value, whereas the size of this subset determines the estimation accuracy. Intuitively, to recover the fMC result, we only need to estimate how each fMC p-value compares with the corresponding BH threshold; hypotheses with p-values far away from the threshold can be estimated less accurately, thus requiring fewer MC samples. \texttt{AMT} turns this pure computational fMC procedure into a statistical estimation problem, where adaptive sampling can be used.

The specific adaptive sampling procedure is developed via a connection to the pure exploration problem in multi-armed bandits (MAB) \cite{audibert2010best,jamieson2014lil}. Specifically, the top-$k$ identification problem \cite{kalyanakrishnan2012pac,chen2017nearly,simchowitz2017simulator} aims to identify the best $k$ arms via adaptive sampling. For \texttt{AMT}, we can think of the $m$ null hypotheses as arms, fMC p-values as arm parameters, and MC samples as observations for each arm. Then recovering the fMC result corresponds to identifying a subset of best arms with small p-values. The difference is that the size of this subset is not known ahead of time --- it is a function of the fMC p-values that needs to be learned from data. Nonetheless, the techniques in MAB is borrowed to develop \texttt{AMT}.

\subsection{Background}
{\bf Permutation test.} Consider testing the association between input $X$ and response $Y$ using $N$ data samples, i.e., the input vector $\mathbf{X} \in \mathbb{R}^N$ and the response vector $\mathbf{Y} \in \mathbb{R}^N$. A reasonable test statistic can be the Pearson's correlation $\rho(\mathbf{X}, \mathbf{Y})$. Let $\sigma$ be a permutation on $\{1, \ldots, N\}$ and $\mathcal{S}$ be the set of all possible permutations. The permutation test statistic by permuting the response with $\sigma$ can be written as $\rho(\mathbf{X}, \mathbf{Y_\sigma})$. Under the null hypothesis that the response $\mathbf{Y}$ is exchangeable among $N$ samples, the rank of the observed test statistic $\rho(\mathbf{X}, \mathbf{Y})$ among all permutation test statistics is uniformly distributed. Hence, the permutation p-value $p^{\text{Perm}} \eqdef \frac{1}{\vert \mathcal{S} \vert} \sum_{\sigma \in \mathcal{S}} \ind\{\rho(\mathbf{X}, \mathbf{Y_\sigma}) \geq \rho(\mathbf{X}, \mathbf{Y})\}$
follows a uniform distribution over the support $\{\frac{1}{\vert \mathcal{S} \vert}, \frac{2}{\vert \mathcal{S} \vert}, \cdots, 1\}$. In most cases, the sample size $N$ is too large for computing all possible permutations; MC permutation test is used where the permutations are uniformly sampled from $\mathcal{S}$.

{\bf FDR control.} For simultaneously testing $m$ null hypotheses with p-values $P_1, \cdots, P_m$, a common goal is to control FDR, defined as the expected proportion of false discoveries 
\begin{align}
    \text{FDR} \eqdef \E \left[ \frac{\text{Number of false discoveries}}{\text{Number of discoveries}}\right].
\end{align}
The most widely-used FDR control algorithm is the BH procedure \cite{benjamini1995controlling}. Let $P_{(i)}$ be the $i$th smallest p-value. The BH procedure rejects hypotheses $P_{(1)}, \cdots, P_{(r^*)}$, where $r^*$ is the critical rank defined as $r^* \eqdef \max \left\{r: P_{(r)} \leq \frac{r}{m} \alpha, r \in \{1,2,\cdots, m\} \right\}$.
The BH procedure controls FDR under the assumption that the null p-values are independent and stochastically greater than the uniform distribution.

\subsection{Related works}
The idea of algorithm acceleration by converting a computational problem into a statistical estimation problem and designing the adaptive sampling procedure via MAB has witnessed a few successes. An early example of such works is the Monte Carlo tree search method \citep{OR_folks,KocSze} to solve large-scale Markov decision problems, a central component of modern game playing systems like AlphaZero \cite{alpha0}. More recent examples include adaptive hyper-parameter tuning for deep neural networks \cite{JamTal,LiJamEtAl} and medoid computation \cite{bagaria2018medoids}. The latter work gives a clear illustration of the power of such an approach. The medoid of a set of $n$ points is the point in the set with the smallest average distance to other points. The work shows that by adaptively \emph{estimating} instead of exactly \emph{computing} the average distance for each point, the computational complexity can be improved from $n^2$ of the naive method to almost linear in $n$. This idea is further generalized in \texttt{AMO} \cite{bagaria2018adaptive} that considers optimizing an arbitrary objective function over a finite set of inputs. In all these works, the adaptive sampling is by standard best-arm identification algorithms. 
This present work also accelerates the fMC procedure by turning it into a statistical estimation problem.
However, no MAB algorithm is readily available for this particular problem. 

Our work applies MAB to FDR control by building an efficient {\em computational} tool to run the BH procedure {\em given the data}.
There are recent works that also apply MAB to FDR control but in a {\em statistical inference} setting where the {\em data collection} process itself can be made adaptive over the different tests. In these works, each arm also corresponds to a test, but each arm parameter takes on a value that corresponds to either null or alternative. Fresh data can be adaptively sampled for each arm and the goal is to select a subset of arms while controlling FDR \cite{yang2017framework,jamieson2018bandit}. In such settings, each observation is a new data and the p-values for the alternative hypotheses can be driven to zero. This is different from \texttt{AMT} where the arm observations are MC samples simulated from the data. As a result, the fMC p-values themselves are the arm parameters and the goal is to {\em compute} them efficiently to perform BH. In an application like GWAS, where all the SNPs data are typically collected simultaneously via whole genome sequencing, adaptive data collection does not apply but overcoming the computational bottleneck of the full MC procedure is an important problem addressed by the present work. See more details of bandit FDR in Supp. Sec. \ref{subsec:bandit_fdr}.

In the broader statistical literature, adaptive procedures \cite{besag1991sequential,gandy2017implementing} or importance sampling methods \cite{yu2011efficient,shi2016efficiently} were developed to efficiently compute a single MC p-value. For testing multiple hypotheses with MC tests, interesting heuristic adaptive algorithms were proposed without formal FDR guarantee \cite{sandve2011sequential,gandy2017quickmmctest}; the latter \cite{gandy2017quickmmctest} was developed via modifying Thompson sampling, another MAB algorithm. Asymptotic results were provided that the output of the adaptive algorithms will converge to the desired set of discoveries \cite{guo2008adaptive,gandy2014mmctest,gandy2016framework}. Specifically, the most recent work \cite{gandy2016framework} provided a general result that incorporates virtually all popular multiple testing procedures. However, none of the above works provide a standard FDR control guarantee (e.g., $\text{FDR}\leq \alpha$) nor an analysis of the MC sample complexity; the MC sample complexity was analyzed in another work only for the case of using Bonferroni procedure \cite{hahn2015optimal}. In the present work, standard FDR control guarantee is provided along with upper and lower bounds on the MC sample complexity, establishing the optimality of \texttt{AMT}.

There are also works on fast MC test for GWAS or eQTL (expression quantitative trait loci) study \cite{pahl2010permory,kimmel2006fast,browning2008presto,jiang2012statistical,zhang2012rapid}; they consider a different goal which is to accelerate the process of separately computing each MC p-value. In contrast, \texttt{AMT} accelerates the entire workflow of both computing MC p-values and applying BH on them, where the decision for each hypothesis also depends globally on others. The state-of-art method is the sequential Monte Carlo procedure (sMC) that is implemented in the popular GWAS package \texttt{PLINK} \cite{besag1991sequential,purcell2007plink,che2014adaptive}. For each hypothesis, it keeps MC sampling until having observed $s$ extreme events or hit the sampling cap $n$. Then BH is applied on the set of sMC p-values. Here we note that the sMC p-values are conservative so this procedure controls FDR. sMC is discussed and thoroughly compared against in the rest of the paper.


\begin{figure*}
    \centerline{
    \subfigure[Initialization]{\includegraphics[width=0.32\textwidth]{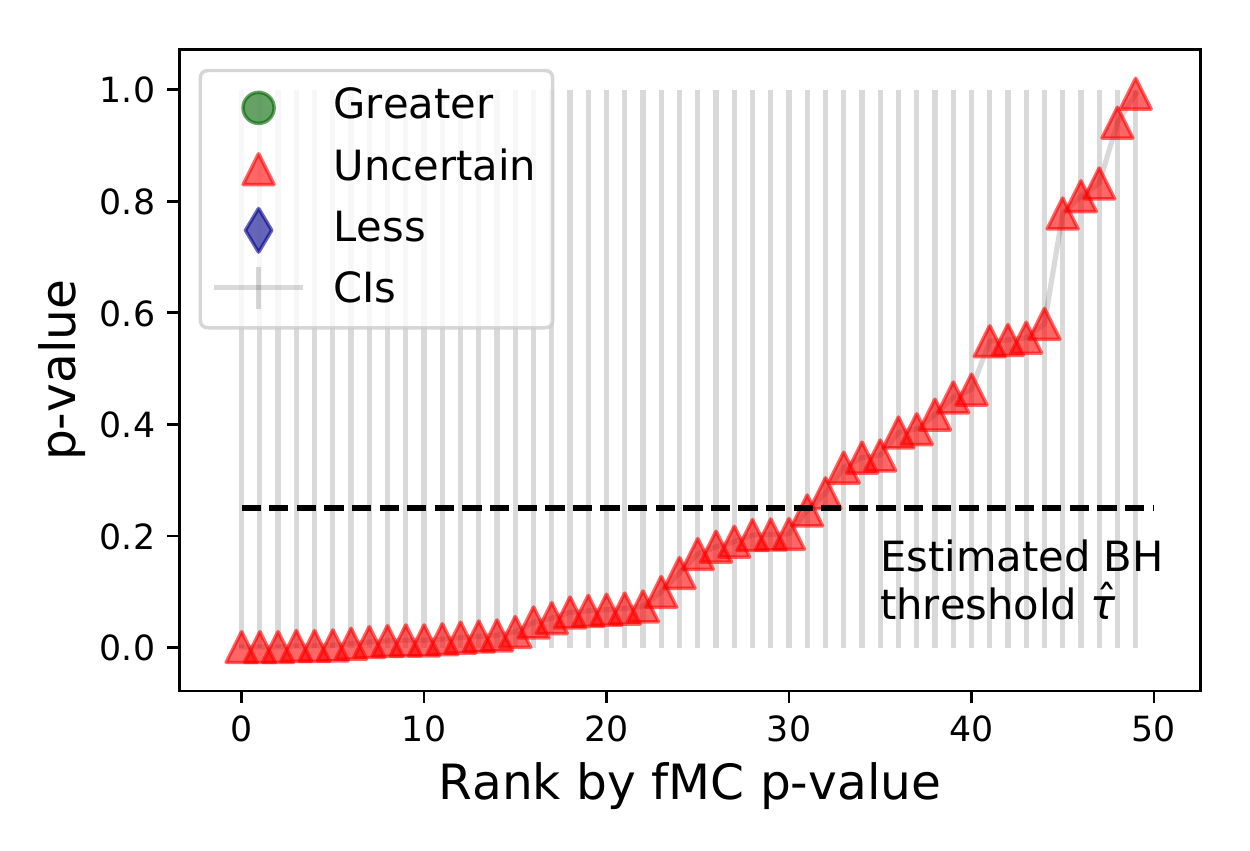}}
    \subfigure[In progress]{\includegraphics[width=0.32\textwidth]{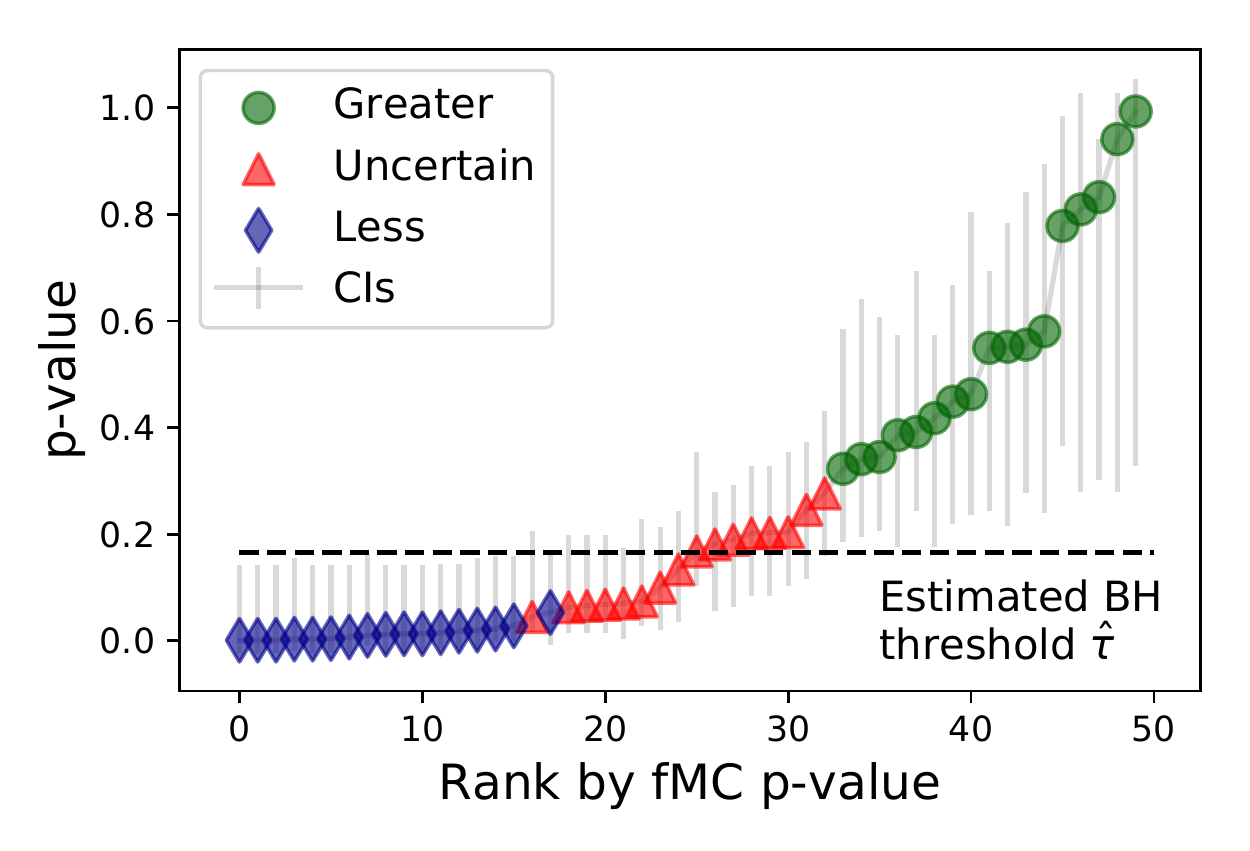}}
    \subfigure[Termination]{\includegraphics[width=0.32\textwidth]{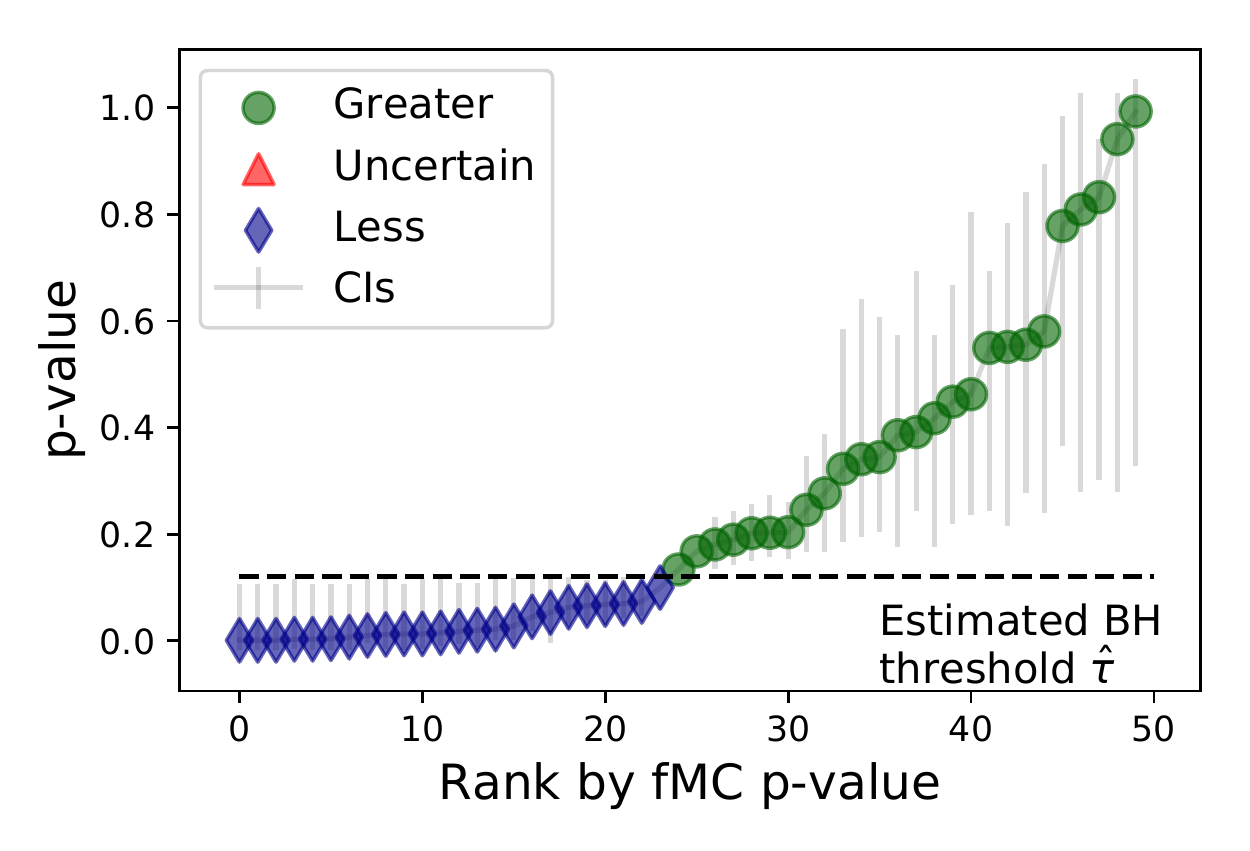}}}
    \caption{Progression of \texttt{AMT}. In this toy example, n=1000, m=50, and $\alpha$=0.25. \texttt{AMT} maintains upper and lower CBs for each hypothesis (vertical grey bar). (a) At initialization, the estimated BH threshold is set to be maximum $\hat{\tau}=\alpha$ while all CBs cross $\hat{\tau}$. Thus, all hypotheses are in $\mathcal{U}$ and need to be further sampled (red triangle).  (b) As the algorithm progresses, more MC samples narrow the confidence intervals and some hypotheses become certain to be greater (green circle) or less (blue diamond) than the estimated BH threshold. The estimated BH threshold also moves down accordingly. (c) At termination, there is no uncertain hypothesis.}
    \label{fig:progression}
\end{figure*}

\section{Problem Formulation}
Let $m$ be the number of hypotheses and $P_1^{\infty}, \cdots, P_m^{\infty}$ be the ideal p-values. We use the standard notation $[m]\eqdef \{1,2,\cdots,m\}$. For two numbers $a,b \in \mathbb{R}$, $a\wedge b$ means $\min(a,b)$ and $a\vee b$ means $\max(a,b)$.

For each hypothesis $i \in [m]$, we assume the MC samples are available of the form  
\begin{align}\label{eq:MC_sample}
    \left[ B_{i,1}, B_{i,2}, \cdots, B_{i,n} \Big\vert P_i^{\infty}=p^{\infty}_i\right] \overset{\text{i.i.d.}}{\sim} \text{Bern} (p^{\infty}_i).
\end{align}
Note that one can think of $B_{i,j} = \ind \{T^{\text{null}}_{i,j} \geq T^{\text{obs}}_i\}$. 

To contrast with adaptive MC sampling, we change the superscript from ``MC($n$)'' to ``fMC'' for the fMC p-values. 
Specifically, the fMC procedure uniformly computes $n$ MC samples for each hypothesis, yielding fMC p-values 
\begin{align}\label{eq:MC_p_val}
    P_i^{\text{fMC}} \eqdef \frac{1}{n+1} \left( 1 + \sum_{j=1}^n B_{i,j}\right),~~~~i\in [m].
\end{align}
Here, the extra ``1'' in the brackets is to make the fMC p-value conservative under the null. 
We would like to point out that there are two sources of randomness. The first is from the data generation process corresponding to the ideal p-values $\{P_i^{\infty}\}$ while the second is from MC sampling; they correspond to panel a and panels b-c in Fig.\ref{fig:schema}, respectively. The second source of randomness corresponding to MC sampling is of primary interest in the present paper.

Applying the BH procedure to the fMC p-values yields a set of discoveries $\mathcal{R}^{\text{fMC}}\subset [m]$. 
Since the fMC p-values are stochastically greater than the uniform distribution under the null hypothesis \cite{phipson2010permutation}, the set of fMC discoveries $\mathcal{R}^{\text{fMC}}$ has a FDR controlled below the nominal level $\alpha$. 
Here, let $P_{(r)}^{\text{fMC}}$ represent the $r$th smallest p-value and define the critical rank as 
\begin{align}\label{eq:fMC_critical_rank}
    r^* \eqdef \max \left\{r: P_{(r)}^{\text{fMC}} \leq \frac{r}{m} \alpha, r \in [m] \right\}.
\end{align}
The BH threshold can be written as $\tau^* \eqdef \frac{r^*}{m}\alpha$ while the set of fMC discoveries $\mathcal{R}^{\text{fMC}} \eqdef \{i: P_i^{\text{fMC}}\leq \tau^*\}$.
The goal is to compute the fMC discoveries $\mathcal{R}^{\text{fMC}}$ with high probability while requiring minimum number of MC samples. 
Formally, we aim to minimize the number of MC samples for the algorithm such that the algorithm output $\mathcal{R} \subset [m]$ satisfies $\P(\mathcal{R} = \mathcal{R}^{\mathrm{fMC}}) \geq 1-\delta$,
for some given $\delta>0$.


\section{Algorithm}

\begin{algorithm}[tb]
   \caption{The \texttt{AMT} algorithm.}
   \label{alg:adapermute}
\begin{algorithmic}
   \STATE {\bfseries Input:} failure probability $\delta$, nominal FDR $\alpha$.
   \STATE {\bfseries Initialization:} $\frac{\delta}{2 mL}$-CBs $\{p^{\text{lb}}_i=0, p^{\text{ub}}_i=1\}_{i\in [m]}$, critical rank estimate $\hat{r}=m$, BH threshold estimate $\hat{\tau}=\alpha$, hypothesis sets $\mathcal{C}_{\text{g}} = \emptyset, \mathcal{C}_{\text{l}} = \emptyset, \mathcal{U} = [m]$. 
   \REPEAT
   \STATE {\bfseries Sample} obtain the next batch of MC samples for each hypothesis in $\mathcal{U}$ and update their $\frac{\delta}{2 mL}$-CBs (Sec. \ref{SubSec:CI}).
   \STATE {\bfseries Update} reduce $\hat{r}$ one at a time and update $\mathcal{C}_{\text{g}}$ correspondingly until the following hold at the same time: 
   \begin{align*}
       \mathcal{C}_{\text{g}} = \left\{i: p^{\text{lb}}_i > \frac{\hat{r}}{m}\alpha \right\},~~~~ \hat{r} = m - \vert \mathcal{C}_{\text{g}}\vert.
   \end{align*}
    Update the estimated BH threshold $\hat{\tau} = \frac{\hat{r}}{m}\alpha$ and the hypothesis sets
   \begin{align*}
       \mathcal{U} = \{i: p^{\text{lb}}_i \leq \hat{\tau} < p^{\text{ub}}_i\},~~~~ \mathcal{C}_{\text{l}} = \{i: p^{\text{ub}}_i \leq \hat{\tau}\}.
   \end{align*}
   \UNTIL{$\mathcal{U} \neq \emptyset$}
   \STATE {\bfseries Return:} $\mathcal{R} = \mathcal{C}_{\text{l}}$.
\end{algorithmic}
\end{algorithm}

\texttt{AMT} is described as in Algorithm \ref{alg:adapermute}.
It adopts a top-down procedure by starting with an initial critical rank estimate $\hat{r}=m$ and gradually moving down until it reaches the true critical rank $r^*$.
Specifically, it maintains upper and lower confidence bounds (CBs) for each hypothesis $(p^{\text{ub}}_i, p^{\text{lb}}_i)$, the critical rank estimate $\hat{r}$, and the corresponding BH threshold estimate $\hat{\tau} = \frac{\hat{r}}{m} \alpha$. Based on the current estimate, the hypotheses can be categorized as:
\begin{equation}\label{eq:h_set}
    \begin{split}
        & \text{Certain to be greater than $\hat{\tau}$: }~~~~\mathcal{C}_{\text{g}} = \{i: p^{\text{lb}}_i > \hat{\tau}\} \\
        & \text{Certain to be less than $\hat{\tau}$: }~~~~\mathcal{C}_{\text{l}} = \{i: p^{\text{ub}}_i \leq \hat{\tau}\} \\
        & \text{Uncertain: }~~~~\mathcal{U} = \{i: p^{\text{lb}}_i \leq \hat{\tau} < p^{\text{ub}}_i\}.
    \end{split}
\end{equation}

As shown in Fig.\ref{fig:progression}a, at initialization the critical rank estimate $\hat{r}$ is set to be the largest possible value $m$ and all hypotheses are uncertain as compared to the estimated BH threshold $\hat{\tau}$; they will be further sampled.
In Fig.\ref{fig:progression}b, as more MC samples narrow the confidence intervals, some hypotheses will become certain to be greater/less than $\hat{\tau}$; they will leave $\mathcal{U}$ and stop being sampled. 
At the same time, according to \eqref{eq:fMC_critical_rank} the estimate $\hat{r}$ cannot be the true critical rank $r^*$ if more than $m-\hat{r}$ p-values are greater than the corresponding estimated BH threshold $\hat{\tau}$. Therefore, we can decrease $\hat{r}$ and update $\mathcal{C}_{\text{g}}$ until $m-\hat{r} = \vert \mathcal{C}_{\text{g}} \vert$. Note that the estimated BH threshold will be reduced correspondingly. 
The algorithm repeats such a sample-and-update step until the set $\mathcal{U}$ becomes empty as shown in Fig.\ref{fig:progression}c. Then it outputs the discoveries.

For practical consideration, every time a batch of MC samples is obtained for each hypothesis in $\mathcal{U}$ instead of one, with batch sizes $[h_1, \cdots, h_L]$ prespecified as $h_l = \gamma^l$ for some $\gamma >1$. Here, $\sum_{l=1}^L h_l = n$ and $L = \Theta(\log n)$. The batched sizes are chosen as a geometric series so that 1) after every batch, the confidence intervals for a hypothesis being sampled will shrink by roughly a constant factor; 2) the number of batches L is relatively small to save the computation on updating the estimated quantities. In the actual implementation, we chose $h_1=100$ and $\gamma=1.1$ for all experiments.

\subsection{Confidence bounds \label{SubSec:CI}}
Since the fMC p-values are themselves random, the CBs are defined conditional on the fMC p-values, where the MC samples for the $i$th hypothesis are drawn \emph{uniformly and without replacement} from the set of all $n$ MC samples $\{b_{i,j}\}_{j\in [n]}$. This gives finite population CBs whose uncertainty is $0$ when $n$ samples are obtained. 

Specifically, for any $k \in [n]$, let $\tilde{B}_1, \tilde{B}_2, \cdots, \tilde{B}_k$ be random variables sampled uniformly and without replacement from the set $\{b_j\}_{j \in [n]}$ and $\hat{p}_k = \frac{1}{k} \left(1 \vee \sum_{i=j}^k \tilde{B}_j\right)$.
The $\delta$-CBs $p^{\text{ub}}, p^{\text{lb}}$ satisfy $\P(p \geq p^{\text{ub}}) \leq \delta$ and $\P(p \leq p^{\text{lb}}) \leq \delta.$
For the analysis, we assume that the CBs take the form 
\begin{align}\label{eq:CB_form}
    p^{\text{ub}} = \hat{p}_k + \sqrt{\frac{\hat{p}_k c(\delta)}{k}},~~~~ p^{\text{lb}} = \hat{p}_k - \sqrt{\frac{\hat{p}_k c(\delta)}{k}}, 
\end{align}
where $c(\delta)$ is a constant depending only on the probability $\delta$.
Eq. \eqref{eq:CB_form} represents a natural form that most CBs satisfy with different $c(\delta)$. We consider this general form to avoid tying \texttt{AMT} up with a specific type of CB, for the adaptive procedure is independent of the choice of CB. All binomial confidence intervals can be used here for this sample-without-replacement case \cite{bardenet2015concentration}; we chose Agresti-Coull confidence interval \cite{agresti1998approximate} for the actual implementation. 



\subsection{Comparison to sMC\label{subsec:smc}}
\begin{figure}
    \centering
    \centerline{\includegraphics[width=0.45\textwidth]{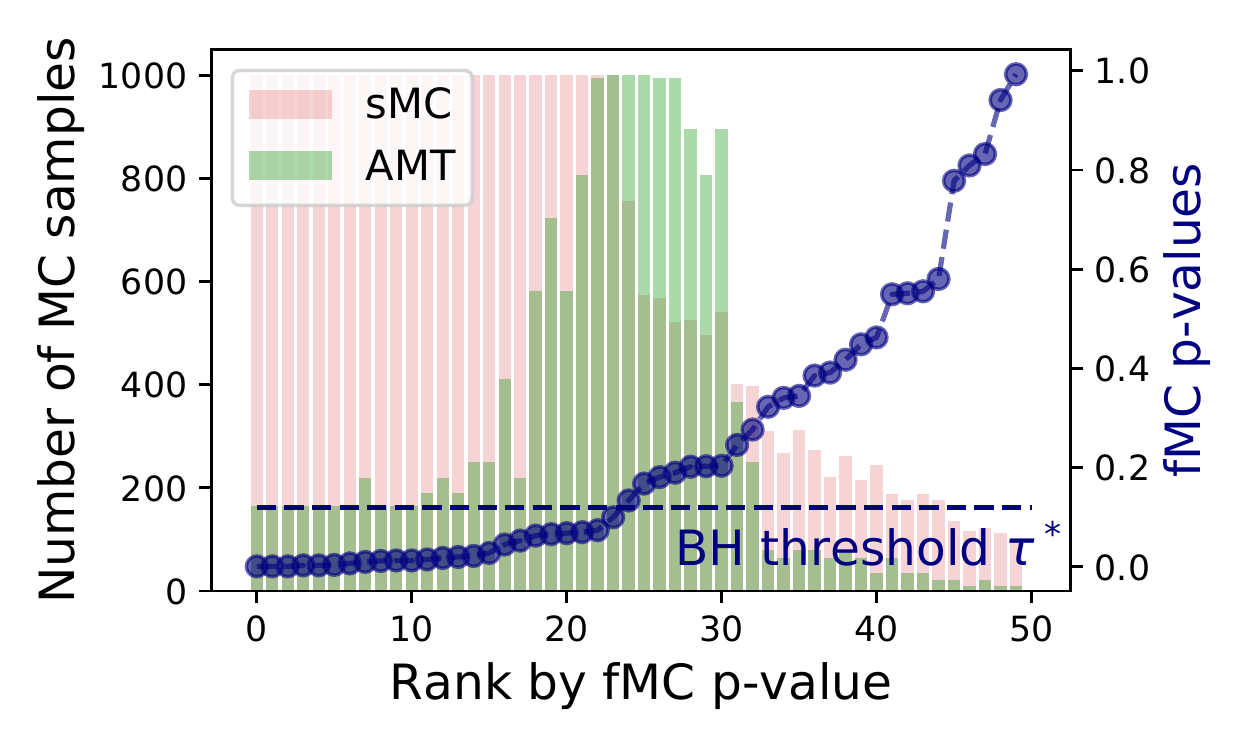}}
    \caption{A toy example with n=1000 and m=50. sMC computes more MC samples for hypotheses with smaller p-values while \texttt{AMT} computes more MC samples for hypotheses with p-values closer to the BH threshold.}
    \label{fig:compare}
\end{figure}

In sMC, for each hypothesis, MC samples are obtained until either $s$ extreme values (MC observation equal to 1) are observed, or $n$ total permutations are computed with $S$ total successes, where $S < s$. Let $K$ be the number of MC samples obtained for the hypothesis. The sMC p-value is defined as 
\begin{align}
    P^{\text{sMC}} \eqdef \left\{ \begin{array}{cc}
    \frac{s}{K}     & K < n \\
    \frac{S+1}{n+1}     & K=n 
    \end{array}\right..
\end{align}
After this, BH is applied on the set of sMC p-values to obtain the testing result.

As shown in Fig.\ref{fig:compare}, sMC computes more MC samples for hypotheses with smaller p-values while \texttt{AMT} computes more MC samples for hypotheses with p-values closer to the BH threshold, effectively addressing the hardness of recovering the fMC result, i.e., deciding how each fMC p-value compares with the BH threshold. See also Supp. Sec. \ref{subsec:sMC_param} for how to choose the parameter $s$.

\section{Theoretical Guarantee}
We present the high probability recovery and FDR control result, the upper bound, and the lower bound in order. 
For the upper bound, we first state the $\tilde{O}(\sqrt{n}m)$ result in Proposition \ref{crly:ub}, which is a direct consequence of the main instance-wise upper bound as stated in Theorem \ref{thrm:instance_ub}.
\subsection{Correctness}
\begin{theorem}(Correctness)\label{thrm:recovery}
    \texttt{AMT} recovers the fMC result with probability at least $1-\delta$, i.e.,
    \begin{align}\label{eq:recovery}
        \P(\mathcal{R}^{\mathrm{AMT}} = \mathcal{R}^{\mathrm{fMC}}) \geq 1-\delta. 
    \end{align}
    Moreover, \texttt{AMT} controls FDR at level $\pi_0 \alpha+\delta$, where $\pi_0$ is the null proportion. 
\end{theorem}

\begin{remark}
A stronger version is actually proved for \eqref{eq:recovery}: \texttt{AMT} recovers the fMC result with probability at least $1-\delta$ conditional on any set of fMC p-values $\{P_i^{\text{fMC}}\} = \{p_i\}$, i.e.,
\begin{align}
    \P \left(\mathcal{R}^{\mathrm{AMT}} = \mathcal{R}^{\mathrm{fMC}} \Big\vert \{P_i^{\text{fMC}}\} = \{p_i\} \right) \geq 1-\delta. 
\end{align} 
This also corresponds to the $\delta$-correctness definition in the lower bound Theorem \ref{thrm:lb}.
For the FDR control argument, $\delta$ is negligible as compared to $\alpha$; $\delta$ is set to be a $o(1)$ term, e.g., $\delta = \frac{1}{m}$. Hence, $\pi_0 \alpha+\delta \leq \alpha$ in most cases. 
\end{remark}

\subsection{Upper bound}
Without loss of generality, let us assume that the ideal p-values, corresponding to the generation of the data, are drawn i.i.d. from an unknown distribution $F(p)$, which can be understood as a mixture of the null distribution and the alternative distribution, i.e., $F(p) = \pi_0 p + (1-\pi_0) F_1(p)$, where $\pi_0$ is the null proportion and $F_1(p)$ is the alternative distribution. The following result shows that the sample complexity of \texttt{AMT} is $\tilde{O}(\sqrt{n}m)$ under mild assumptions of $F(p)$. 

\begin{proposition}\label{crly:ub}
Assume that the ideal p-values are drawn i.i.d. from some unknown distribution $F(p)$ with density $f(p)$ that is either constant ($f(p)=1$) or continuous and monotonically decreasing. 
With $\delta = \frac{1}{m\sqrt{n}}$, the total number of MC samples for \texttt{AMT} satisfies
\begin{align}
    \E[N] = \tilde{O}(\sqrt{n}m),
\end{align}
where $\tilde{O}$ hides logarithmic factors with respect to $m$ and $n$.
\end{proposition}

\begin{remark}
The asymptotic regime is when $m\rightarrow\infty$ while $n=\Omega(m)$. This is because the number of MC samples $n$ should always be larger than the number of hypothesis tests $m$.
A more complete result including $\delta$ is $\tilde{O}\left(\sqrt{n}m \log\frac{1}{\delta} + \delta mn\right)$.

For the assumption on the ideal p-value distribution $F(p)$, $f(p)=1$ corresponds to the case where all hypotheses are true null while $f(p)$ being continuous and monotonically decreasing essentially assumes that the alternative p-values are stochastically smaller than uniform. 
Such assumption includes many common cases, e.g., when the p-value is calculated from the z-score $Z_i \sim \mathcal{N}(\mu, 1)$ with $\mu=0$ under the null and $\mu>0$ under the alternative \cite{hung1997behavior}.

A strictly weaker but less natural assumption is sufficient for the $\tilde{O}(\sqrt{n}m)$ result. Let $\tau^\infty = \sup_{[0,1]}\{\tau: \tau \leq F(\tau) \alpha\}$. It assumes that $\exists c_0, c_1 > 0$ s.t. $\forall p \in [\tau^\infty-c_0, 1]$, $f(p)\leq \frac{1}{\alpha}-c_1$. As shown in the proof, $\tau^\infty$ is the BH threshold in the limiting case and $f(\tau^\infty) < \frac{1}{\alpha}$ as long as $f(p)$ is strictly decreasing on $[0,\tau^\infty]$. Hence, this weaker assumption contains most practical cases and the $\tilde{O}(\sqrt{n}m)$ result holds generally. However, this weaker assumption involves the definition of $\tau^\infty$ which is technical. We therefore chose the stronger but more natural assumption in the statement of the corollary.
\end{remark}

Proposition \ref{crly:ub} is based on an instance-wise upper bound conditional on the fMC p-values $\{P_i^{\text{fMC}}\}= \{p_i\}$, stated as follows.
\begin{theorem} \label{thrm:instance_ub}
Conditioning on any set of fMC p-values $\{P_i^{\mathrm{fMC}}\} = \{p_i\}$, let $p_{(i)}$ be the $i$th smallest p-value and $\Delta_{(i)} = \vert p_{(i)} - \frac{i \vee r^*}{m} \alpha \vert$.
For the CBs satisfying \eqref{eq:CB_form}, the total number of MC samples $N$ satisfies 
\begin{equation*} \label{eq:thrm:bh}
    \begin{split}
        & \E \left[N \Big\vert \{P_i^{\mathrm{fMC}}\} = \{p_i\}\right] \leq  \sum_{i=1}^{r^*} n \wedge \left( \frac{4(1+\gamma)^2 c\left(\frac{\delta}{2 mL}\right)  \tau^*}{\Delta_{(i)}^2}\right) \\
    & + \sum_{i=r^*+1}^m n \wedge \left( \max_{k\geq i} \frac{4(1+\gamma) c\left(\frac{\delta}{2mL}\right) p_{(k)} }{\Delta_{(k)}^2} \right)  + \delta mn.
    \end{split}
\end{equation*}
\end{theorem}

\begin{remark}
Note that $L = \log_{\gamma}n$ and for common CBs, $c(\delta) = \log \frac{1}{\delta}$. By setting $\delta = \frac{1}{m}$ and $\gamma=1.1$, we have
\begin{align*}
    & \E \left[N \Big\vert \{P_i^{\mathrm{fMC}}\} = \{p_i\}\right] \leq  \sum_{i=1}^{r^*} n \wedge \left( \frac{18 \log (50 m^2\log n) \tau^*}{\Delta_{(i)}^2}\right) \\
    & + \sum_{i=r^*+1}^m n \wedge \left( \max_{k\geq i} \frac{9 \log (50 m^2\log n) p_{(k)} }{\Delta_{(k)}^2} \right)  + n.
\end{align*}
The terms in the summations correspond to the number of MC samples for each hypothesis test. The denominator $\Delta_{(i)}^2$ represents the hardness for determining if to reject each hypothesis while the hypothesis-dependent numerator ($\tau^*$ in the first summation and $p_{(k)}$ in the second) represents a natural scaling of the binomial proportion confidence bound. The $\max$ in the second term corresponds to the specific behavior of the top-down approach; it is easy to construct examples where this is necessary. 
The factor $\log (50 m^2\log n)$ corresponds to the high probability bound which is $log$ in $m$ and $loglog$ in $n$. This is preferable since $n$ may be much larger than $m$. 
Overall, the bound is conjectured to be tight except improvements on the $\log m$ term (to perhaps $\log\log m$). 
\end{remark}

\subsection{Lower bound}
We provide a matching lower bound for the $\tilde{O}(\sqrt{n}m)$ upper bound. Here, we define a $\delta$-correct algorithm to be one that, conditional on any set of fMC p-values $\{P_i^{\text{fMC}}\} = \{p_i\}$, recovers the fMC result with probability at least $1-\delta$.
\begin{theorem}\label{thrm:lb}
Assume that the ideal p-values are drawn i.i.d. from some unknown distribution $F(p)$ with null proportion $\pi_0>0$.
$\exists \delta_0>0$, s.t. $\forall \delta<\delta_0$, any $\delta$-correct algorithm satisfies
\begin{align}
    \E[N] = \tilde{\Omega}(\sqrt{n}m),
\end{align}
where $\tilde{\Omega}$ hides logarithmic factors with respect to $m$ and $n$.
\end{theorem}

\begin{remark}
In practical settings most hypotheses are true null and therefore $\pi_0>0$.
\end{remark}

\section{Empirical Results}
\begin{table}
\caption{Recovery of the fMC result.}
\label{tab:simu_recovery}
\vskip 0.15in
\begin{center}
\begin{small}
\begin{tabular}{lcccr}
\toprule
Failure & Avg. MC samples & Prop. of  \\
prob. $\delta$ & per hypothesis ($\pm$std) & success recovery  \\
\midrule
0.001  & 1128$\pm$73 & 100$\%$ \\
0.01   & 1033$\pm$72 & 100$\%$ \\
0.1    & 930$\pm$70 & 100$\%$ \\
\bottomrule
\end{tabular}
\end{small}
\end{center}
\vskip -0.1in
\end{table}

\begin{figure}
    \centering
    \centerline{\includegraphics[width=0.4\textwidth]{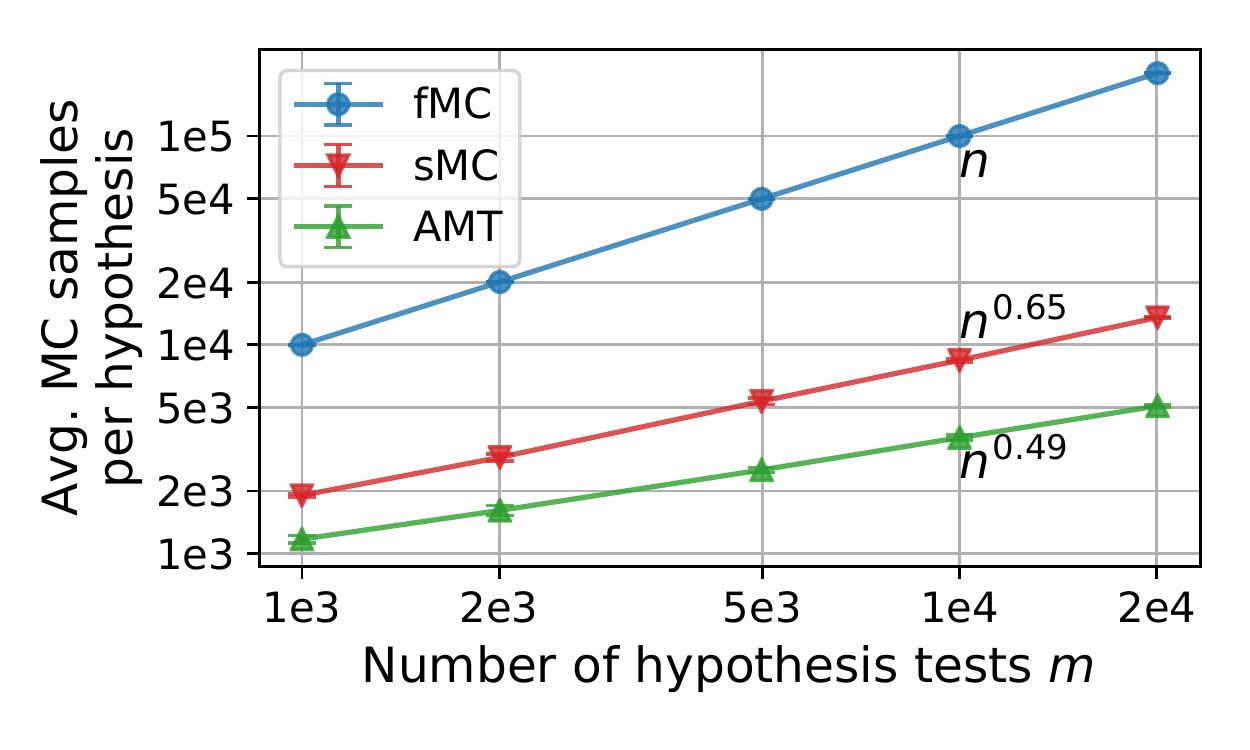}}
    \caption{Average number of MC samples per hypothesis test for different algorithms while increasing the number of hypothesis tests $m$ and letting $n$=10$m$.}
    \label{fig:simu_scale}
\end{figure}

\begin{figure*}[htb!]
    \centerline{
    \subfigure[Nominal FDR $\alpha$.]{\includegraphics[width=0.32\textwidth]{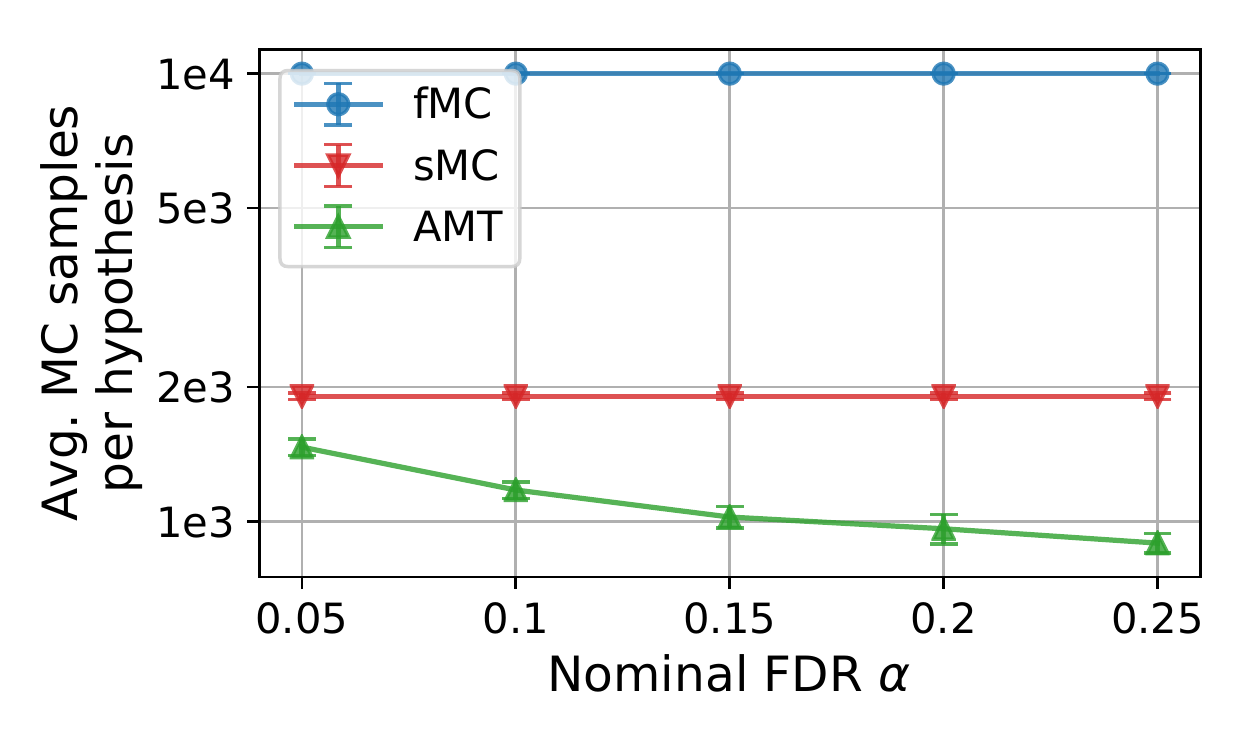}}
    \subfigure[Alternative proportion.]{\includegraphics[width=0.32\textwidth]{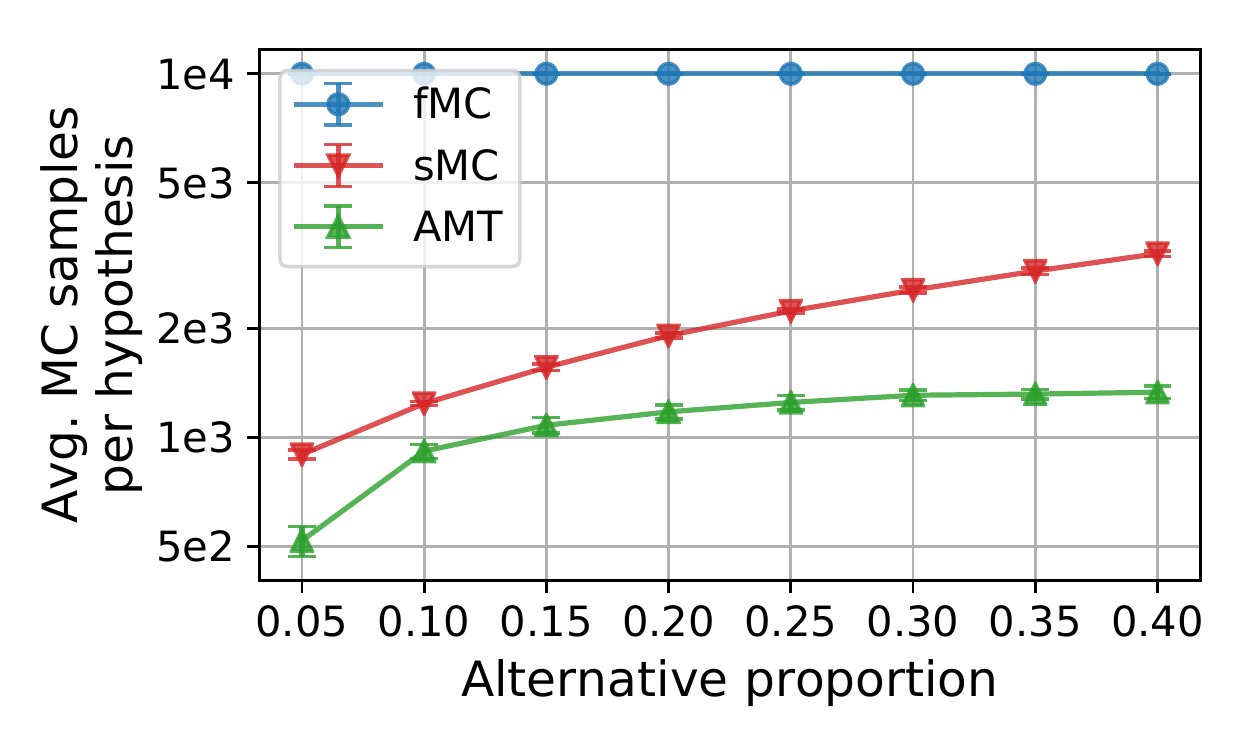}}
    \subfigure[Effect size $\mu$.]{\includegraphics[width=0.32\textwidth]{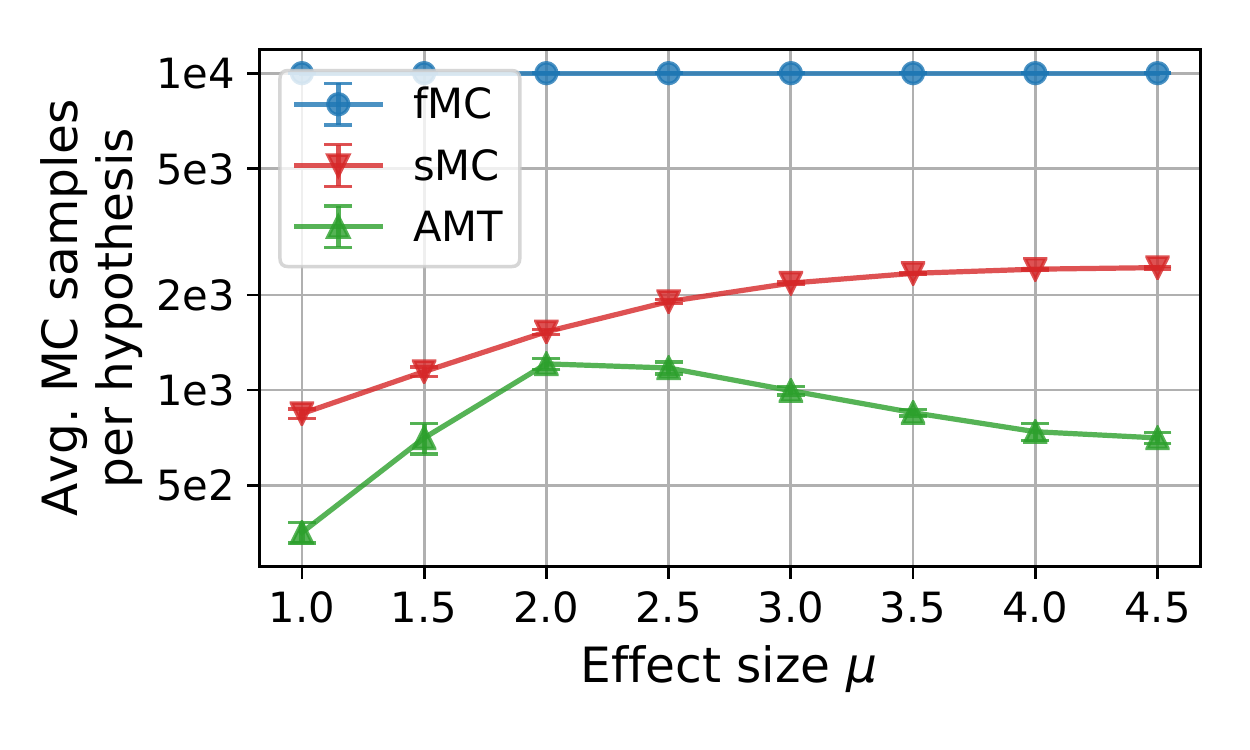}}}
    \caption{Average number of MC samples per hypothesis test for different algorithms while varying different parameters.}
    \label{fig:simu_param}
\end{figure*}
\subsection{Simulated data}
{\bf Setting.} In the default setting, we consider $m$=1000 hypothesis tests, out of which 200 are true alternatives. The p-values are generated from z-scores $Z_i \sim \mathcal{N}(\mu,1)$, where the effect size $\mu$=0 under the null and $\mu$=2.5 under the alternative. The number of fMC samples per hypothesis is set to be $n$=10,000 while the nominal FDR is $\alpha$=0.1. We investigate the performance of \texttt{AMT} by varying different parameters. The performance of sMC is also reported for comparison, where we set its parameter $s$=100 according to the discussion which we postpone to Supp. Sec. \ref{subsec:sMC_param}. We also tested $s$=50, which shows a similar result and is hence omitted.

{\bf Reliability.} We first investigate the reliability of \texttt{AMT} by varying $\delta$, upper bound of the failure probability, where the probability of the CBs is set to be $\frac{\delta}{2 m \log n}$.
For each value of $\delta$ the experiment is repeated 10,000 times. In each repetition, a different set of data is generated and the \texttt{AMT} result is compared to the fMC result while fixing the random seed for MC sampling. As shown in Table \ref{tab:simu_recovery}, \texttt{AMT} recovers the fMC result in all cases while having a 10x gain in sample efficiency. We also found the \texttt{AMT} is rather stable that in practice a larger value of $\delta$ may be used; empirically, \texttt{AMT} starts to fail to recover the fMC result when $\delta$ exceeds $100$. 

{\bf Scaling.} Next we investigate the asymptotic MC sample complexity by increasing $m$ while fixing $n$=$10m$. The experiment is repeated 5 times for each parameter and 95$\%$ confidence intervals are provided. The result is shown in Fig.\ref{fig:simu_scale} where the number of MC samples per hypothesis scales sub-linearly with $n$. A simple linear fitting shows that the empirical scaling for \texttt{AMT} is $n^{0.49}$, validating the $\tilde{O}(\sqrt{n})$ scaling of average MC samples per hypothesis test as derived in Proposition \ref{crly:ub} and Theorem \ref{thrm:lb}. Empirically, \text{sMC} scales sub-linearly but with a higher rate of $n^{0.65}$.

{\bf Varying other parameters.} Finally we vary other parameters including the nominal FDR $\alpha$, alternative proportion, and the effect size $\mu$, where the experiment for each parameter setting is repeated 5 times and 95$\%$ confidence intervals are provided. The results are shown in Fig.\ref{fig:simu_param}. Here, sMC processes each hypothesis separately and computes more MC samples for hypotheses with smaller p-values, as discussed in Sec. \ref{subsec:smc}. However, to obtain the BH result, the true difficulty is quantified by the closeness of the p-values to the BH threshold but zero --- in other words, if a hypothesis test has a very small p-value while the BH threshold is large, it should not be hard to infer that this null hypothesis should be rejected. \texttt{AMT} captures this by adaptively computing more MC samples for null hypotheses with p-values closer to the BH threshold but zero, effectively adapting to different parameter settings and outperforms sMC in terms of the MC sample complexity.

\subsection{GWAS on Parkinson's disease}
We consider a GWAS dataset that aims to identify genetic variants associated with Parkinson's disease \cite{fung2006genome}, which is known to be a complex disease and is likely to be associated with many different SNPs \cite{chang2017meta}; FDR control via BH may yield more new discoveries that are interesting to the community. The dataset comprises 267 cases and 271 controls, each with genotype of 448,001 SNPs that are carefully designed to represent information about several million common genetic variants throughout the genome \cite{international2003international}. 

The phenotype is binary disease/healthy while the genotype is categorical AA/Aa/aa/missing. SNPs with more than $5\%$ missing values are removed to prevent discoveries due to the missing value pattern; this leaves 404,164 SNPs. The MC samples are based on the permutation test using the Pearson's Chi-squared test, where the phenotype is randomly permuted while keeping the same number of cases and controls. This experiment is run on 32 cores (AMD $\text{Opteron}^\text{\texttrademark}$ Processor 6378).


{\bf Small data.} We first compare \texttt{AMT} with fMC on a smaller-scale data that consists of all 23,915 SNPs on chromosome 4 since fMC can not scale to the whole genome.
The number of fMC samples is chosen to be $n$=$250,000$, yielding a total number of $6\times10^9$ MC samples that takes 34 mins to compute with 32 cores (4th row in Table \ref{tab:small_GWAS_time}).
Most fMC p-values are similar to the p-values reported in the original paper \cite{fung2006genome} (Supp. Table \ref{tab:small_GWAS}). The slight difference is because the p-values in the original paper were computed using a different test (Pearson’s Chi-squared test).
FDR level $\alpha$=$0.1$ yields 47 discoveries including all discoveries on chromosome 4 reported in the original paper; $\alpha$=$0.05$ yields 25 discoveries. The \texttt{AMT} result is identical to the fMC result; it takes 123s and an average of 1,241 MC samples per hypothesis, representing a 17x gain in running time and 201x gain in MC sample efficiency. 
The same experiment is performed on other chromosomes (chromosome 1-3), which gives a similar result --- \texttt{AMT} recovers the fMC result in all cases and as shown in Table \ref{tab:small_GWAS_time}, \texttt{AMT} has a gain of 17-39x in running time and 201-314x in MC sample efficiency. See also Supp. Table \ref{supptab:small_GWAS_p_val}.  for the fMC p-values.

{\bf Full data.} We next consider the full dataset with 404,164 SNPs and set the number of fMC samples to be $n$=40,416,400, yielding a total number of $1.6\times 10^{13}$ MC samples. Since there is no internal adaptivity in the fMC procedure, it is reasonable to assume its running time to be proportional to the total number of MC samples, yielding an estimate of 2 months.
It is noted that due to the computational cost, no full-scale permutation analysis has been performed on the dataset. The original paper performed permutation test on a subset of SNPs with theoretical p-values less than 0.05. However, such practice may cause outfitting since the same data is used for both hypothesis selection and testing. 

We run \texttt{AMT} on this dataset with FDR level $\alpha$=0.1, taking 1.1hr to finish and average 13,723 MC samples, representing a gain of 1500x in running time and 3000x in MC sample efficiency. We note that we should expect more computational gain for larger-scale problems since \texttt{AMT} scales linearly with $\sqrt{n}$ while fMC scales linearly with $n$.
In addition, for larger-scale problems the MC samples are computed in larger batches which is more efficient, effectively closing the gap between the gain in actual running time and the gain in MC sample efficiency.

With a FDR level $\alpha$=0.1, \texttt{AMT} made 304 discoveries, including 22/25 SNPs reported in the original paper. Among the three SNPs that are missing, rs355477 ($p^{\text{ub}}$=9.1e-5) and rs355464 ($p^{\text{ub}}$=1.8e-4) are borderline while rs11090762 ($p^{\text{lb}}$=5.9e-2) is likely to be a false positive. 
\texttt{AMT} has a different number of discoveries from the original paper since the original paper reports all SNPs with p-values $<1e$-$4$ as discoveries instead of using the BH procedure. Also, we have not shown that the \texttt{AMT} discoveries are the same as the fMC discoveries here; we validate the correctness of \texttt{AMT} via the aforementioned small data experiment.


{\bf Code availability.} The software is available at 
\\\url{https://github.com/martinjzhang/AMT}

\begin{table}[t]
\caption{Small GWAS data. Average MC samples per hypothesis and running time for fMC and $\texttt{AMT}$. The same experiment is performed on chromosome 1-4 separately. }
\label{tab:small_GWAS_time}
\vskip 0.15in
\begin{center}
\begin{small}
\begin{tabular}{c|cc|cc}
\toprule
\multicolumn{1}{c}{Chromosome} & \multicolumn{2}{|c}{Avg. MC samples} & \multicolumn{2}{|c}{Running time (s)} \\
\cline{2-5}
(\# of SNPs) & fMC & \texttt{AMT} & fMC & \texttt{AMT} \\
\midrule
1 (31,164) & 250,000 & 874 (286x) & 3,148 & 100(31x)\\
2 (32,356) & 250,000 & 797 (314x) & 3,505 & 90 (39x)\\
3 (27,386) & 250,000 & 964 (259x) & 2,505 & 89 (28x)\\
4 (23,915) & 250,000 & 1,241 (201x) & 2,031 & 123 (17x)\\
\bottomrule
\end{tabular}
\end{small}
\end{center}
\vskip -0.1in
\end{table}

\clearpage
\section*{Acknowledgements}
We would like to thank Trevor Hastie, Jennifer Listgarten, Hantian Zhang, Vivek Bagaria, Eugene Katsevich, Govinda Kamath, Tavor Baharav, and Wesley Tansey for helpful discussions and suggestions. 
MZ is partially supported by Stanford Graduate Fellowship. JZ is supported by the Chan-Zuckerberg Initiative, NSF Grant CRII 1657155, NSF Grant AF 1763191, and NIH CEGS. MZ and DT are partially supported by the Stanford Data Science Initiative, NSF Grant CCF 0939370 (Center for Science of Information), NSF Grant CIF 1563098, and NIH Grant R01HG008164.
\bibliography{example_paper}
\bibliographystyle{icml2019}





\newpage

\clearpage
\begin{center}
\textbf{\large Supplemental Materials }
\end{center}
\setcounter{section}{0}
\setcounter{table}{0}
The supplementary material is organized as follows. First we provide additional empirical results and discussions in Supp. Section \ref{sec:supp_result} and Supp. Section \ref{sec:supp_discussions} respectively. Next we present the technical proofs. Specifically, the correctness result Theorem \ref{thrm:recovery} is proved in Supp. Section \ref{suppsec:thrm_recovery}. The instance-wise upper bound Theorem \ref{thrm:instance_ub} is proved in Supp. Section \ref{suppsec:thrm_instance_ub} while the $\tilde{O}(\sqrt{n}m)$ upper bound Proposition \ref{crly:ub} is proved in Supp. Section \ref{suppsec:crly_ub}. The lower bound Theorem \ref{thrm:lb} is proved in Supp. Section \ref{suppsec:thrm_lb}. Finally, the auxiliary lemmas are in Supp. Section \ref{sec:lm}.


\section{Additional Results \label{sec:supp_result}}

\begin{table}[htb!]
\caption{Small GWAS on chromosome 4. }
\label{tab:small_GWAS}
\vskip 0.15in
\begin{center}
\begin{small}
\begin{tabular}{lllll}
\toprule
dbSNP ID & Original & fMC & Rej. at & Rej. at \\
 & p-value & p-value & $\alpha$=0.1 & $\alpha$=0.05 \\
\midrule
rs2242330 & 1.7e-6 & 8.0e-6 & $\surd$ & $\surd$\\
rs6826751 & 2.1e-6 & 1.6e-5 & $\surd$ & $\surd$\\
rs4862792 & 3.5e-5 & 4.0e-6 & $\surd$ & $\surd$\\
rs3775866 & 4.6e-5 & 3.6e-5 & $\surd$ & $\surd$\\
rs355477 & 7.9e-5 & 8.0e-5 & $\surd$ & $\times$\\
rs355461 & 8.3e-5 & 8.0e-5 & $\surd$ & $\times$\\
rs355506 & 8.3e-5 & 8.0e-6 & $\surd$ & $\times$\\
rs355464 & 8.9e-5 & 1.3e-4 & $\surd$ & $\times$\\
rs1497430 & 9.7e-5 & 5.2e-5 & $\surd$ & $\surd$\\
rs11946612 & 9.7e-5 & 4.8e-5 & $\surd$ & $\surd$\\
\bottomrule
\end{tabular}
\end{small}
\end{center}
\vskip -0.1in
\end{table}

\begin{table}[htb!]
\caption{Small GWAS on chromosome 1-3 (There is no discovery reported on chromosomes 2-3 from the orignal paper). }
\label{supptab:small_GWAS_p_val}
\vskip 0.15in
\begin{center}
\begin{small}
\begin{tabular}{lcllc}
\toprule
dbSNP ID & Chromosome & Original & fMC & Rej. at \\
 & & p-value & p-value & $\alpha$=0.1 \\
\midrule
rs988421 & 1 & 4.9e-5 & 3.6e-5 & $\surd$\\
rs1887279 & 1 & 5.7e-5 & 4.4e-5 & $\surd$\\
rs2986574 & 1 & 6.3e-5 & 4.4e-5 & $\surd$\\
rs3010040 & 1 & 8.0e-5 & 6.0e-5 & $\surd$\\
rs2296713 & 1 & 8.0e-5 & 6.0e-5 & $\surd$\\
\bottomrule
\end{tabular}
\end{small}
\end{center}
\vskip -0.1in
\end{table}

\section{Additional Discussions \label{sec:supp_discussions}}
\subsection{Choosing the parameter for sMC\label{subsec:sMC_param}}
For sMC the parameter $s$ need to be chosen \textit{a priori}. A back-of-the-envelope calculation shows that for a hypothesis test with the ideal p-value $p^{\infty}$, the sMC p-value is around $p^{\infty} \pm \frac{p^{\infty}}{\sqrt{s}}$ while the fMC p-value is around $p^{\infty} \pm \sqrt{\frac{p^{\infty}}{n}}$. Suppose the BH threshold on the ideal p-values is $\tau^{\infty}$. Since it is desirable for the BH result on the MC p-values (sMC, fMC) to be close to the BH result on the ideal p-values, the accuracy of the MC p-values with corresponding ideal p-values close to $\tau^{\infty}$ can be thought of as the accuracy of the entire multiple testing problem. Matching such accuracy for sMC and fMC gives that $s = \tau^\infty n = \frac{r^\infty}{m} \alpha n$. When $n$=$10m$ and $\alpha$=$0.1$, we have that $s$=$r^\infty$. That is, $s$ should be at least $100$ if there are more than $100$ discoveries on the ideal p-values. However, since we do not know $r^{\infty}$ before running the experiment, a larger value is preferred. It is noted that values $s$=30-120 are recommended in a recent work \cite{thulin2014cost}.

\subsection{Comparison to bandit FDR\label{subsec:bandit_fdr}}
In the bandit FDR setting \cite{jamieson2018bandit}, each arm has a parameter $\mu_i$ with $\mu_i = \mu_0$ for null arms and $\mu_i>\mu_0+\Delta$ for alternative arms, for some $\mu_0$ and $\Delta>0$ given before the experiment. For arm $i$, i.i.d. observations are available that are bounded and have expected value $\mu_i$. The goal is to select a subset of arms and the selected set should control FDR while achieving a certain level of power. 

Both bandit FDR and \texttt{AMT} aim to select a subset of ``good arms'' as defined by comparing the arm parameters to a threshold. In bandit FDR this threshold is given as $\mu_0$. In \texttt{AMT}, however, this is the BH threshold that is not known ahead of time and needs to be learned from the observed data. The two frameworks also differ in the error criterion. Bandit FDR considers FDR and power for the selected set, a novel criterion in MAB literature. \texttt{AMT}, on the other hand, adopts the traditional PAC-learning criterion of recovering the fMC discoveries with high probability. These distinctions lead to different algorithms: bandit FDR uses an algorithm similar to thresholding MAB \cite{locatelli2016optimal} but with carefully designed confidence bounds to control FDR; \texttt{AMT} devises a new LUCB (lower and upper confidence bound) algorithm that adaptively estimates two things simultaneously: the BH threshold and how each arm compares to the threshold. 

\subsection{Future works}
We have shown that \texttt{AMT} improves the computational efficiency of the fMC workflow, i.e., applying BH on the fMC p-values. A direct extension is to the workflow of applying the Storey-BH procedure \cite{storey2004strong} on the fMC p-values. In addition, in many cases, especially in genetic research, additional covariate information is available for each null hypothesis, e.g., functional annotations of the SNPs in GWAS, where a covariate-dependent rejection threshold can be used to increase testing power \cite{xia2017neuralfdr,zhang2018adafdr}. Extending \texttt{AMT} to such cases would allow both efficient computation of MC p-values and increased power via covariate-adaptive thresholding. Last but not least, MC sampling is an important building block in some modern multiple testing approaches like the model-X knockoff \cite{candes2018panning} or the conditional permutation test \cite{berrett2018conditional}, where ideas in the present paper may be used to improve the computational efficiency. 

\section{Proof of Theorem \ref{thrm:recovery}\label{suppsec:thrm_recovery}}
\begin{proof} (Proof of Theorem \ref{thrm:recovery})
To show \eqref{eq:recovery}, it suffices to show that conditional on any set of fMC p-values $\{P_i^{\text{fMC}}\} = \{p_i\}$, 
\begin{align}\label{eq:recovery_pf1}
    \P \left(\mathcal{R}^{\mathrm{AMT}} = \mathcal{R}^{\mathrm{fMC}} \Big\vert \{P_i^{\text{fMC}}\} = \{p_i\} \right) \geq 1-\delta. 
\end{align}
Let $\mathcal{E}$ denote the event that all CBs hold. Since the number of CBs is at most $2mL$ and each of them holds with probability at least $1-\frac{\delta}{2mL}$ conditional on the fMC p-values, by union bound,
\begin{align*}
    \P\left(\mathcal{E} \Big\vert \{P_i^{\text{fMC}}\} = \{p_i\}\right) \geq 1 - \delta.
\end{align*}

Next we show that $\mathcal{E}$ implies $\mathcal{R}^{\texttt{AMT}}=\mathcal{R}^{\mathrm{fMC}}$, which further gives \eqref{eq:recovery_pf1}. Let $T$ be the total number of rounds, which is finite since at most $mn$ MC samples will be computed. For any round $t$, let ``$(t)$'' represent the corresponding values before the MC sampling of the round, e.g., $\hat{r}(t)$, $\hat{\tau}(t)$, $\mathcal{C}_\text{g}(t)$, $\mathcal{C}_\text{l}(t)$, $\mathcal{U}(t)$. Also, let $(T+1)$ represent the values at termination. 
For any $t\in[T+1]$, 
\begin{enumerate}
    \item if $\hat{r}(t) > r^*$, by \eqref{eq:fMC_critical_rank} more than $m - \hat{r}(t)$ fMC p-values are greater than $\hat{\tau}(t)$ whereas $\vert \mathcal{C}_\text{g}(t) \vert = m - \hat{r}(t)$. Thus, there is at least one hypothesis that has fMC p-value greater than $\hat{\tau}(t)$ and is not in $\mathcal{C}_\text{g}(t)$. On $\mathcal{E}$, it cannot be in $\mathcal{C}_{\textrm{l}}(t)$. Hence, it is in $\mathcal{U}(t)$, giving that $\mathcal{U}(t) \neq \emptyset$. Thus, $t\neq T+1$ and the algorithm will not terminate. 
    \item if $\hat{r}(t) = r^*$, there are $m-r^*$ hypotheses in $\mathcal{C}_\text{g}(t)$ corresponding to those with fMC p-values greater than $\tau^*$. Other hypotheses all have fMC p-values less than $\tau^*$ and hence, on $\mathcal{E}$, will not enter $\mathcal{C}_\text{g}$ after further sampling. Therefore, $\hat{r}(t)$ will not further decrease. 
\end{enumerate}
Therefore, $\hat{r}(T+1) =  r^*$. Since $\mathcal{U}(T+1)=\emptyset$, on $\mathcal{E}$, $\mathcal{C}_\text{l}(T+1)$ contains all hypotheses with fMC p-values less than $\tau^*$, i.e., $\mathcal{C}_\text{l}(T+1) = \mathcal{R}^{\mathrm{fMC}}$. Hence, we have shown \eqref{eq:recovery_pf1}.

Next we prove FDR control. Let $\text{FDP}(\mathcal{R}^{\mathrm{fMC}})$ and $\text{FDR}(\mathcal{R}^{\mathrm{fMC}})$ denote the false discovery proportion and FDR of the set $\mathcal{R}^{\mathrm{fMC}}$, respectively. It is noted that $\text{FDR}(\mathcal{R}^{\mathrm{fMC}}) = \E[\text{FDP}(\mathcal{R}^{\mathrm{fMC}})]$. Let $\mathcal{E}_1$ denote the event that $\mathcal{R}^{\mathrm{AMT}} = \mathcal{R}^{\mathrm{fMC}}$ and $\mathcal{E}_1^c$ be the complement of $\mathcal{E}_1$. Then $\P(\mathcal{E}_1^c) \leq \delta$ due to \eqref{eq:recovery} that we have just proved. For \texttt{AMT}, 
\begin{align}
    & \text{FDR}(\mathcal{R}^{\texttt{AMT}}) = \E[\text{FDP}(\mathcal{R}^{\texttt{AMT}})] \\
    & = \E[\text{FDP}(\mathcal{R}^{\texttt{AMT}}) \vert \mathcal{E}_1] \P(\mathcal{E}_1) + \E[\text{FDP}(\mathcal{R}^{\texttt{AMT}}) \vert \mathcal{E}_1^c] \P(\mathcal{E}_1^c). \label{eq:recovery_pf2}
\end{align}
The first term of \eqref{eq:recovery_pf2}
\begin{align*}
    &\E[\text{FDP}(\mathcal{R}^{\texttt{AMT}}) \vert \mathcal{E}_1] \P(\mathcal{E}_1) = \E[\text{FDP}(\mathcal{R}^{\mathrm{fMC}}) \vert \mathcal{E}_1] \P(\mathcal{E}_1) \\&\leq \E[\text{FDP}(\mathcal{R}^{\mathrm{fMC}})] = \text{FDR}(\mathcal{R}^{\mathrm{fMC}}) \leq \pi_0 \alpha,
\end{align*}
where the last inequality is because the fMC p-values are stochastically greater than the uniform distribution under the null hypothesis, and hence, applying BH on them controls FDR at level $\pi_0 \alpha$.

The second term of \eqref{eq:recovery_pf2} is upper bounded by $\delta$ as FDP is always no greater than $1$.
Therefore, 
\begin{align*}
    \text{FDR}(\mathcal{R}^{\texttt{AMT}}) \leq \pi_0\alpha + \delta.
\end{align*}
\end{proof}

\section{Proof of Theorem \ref{thrm:instance_ub}\label{suppsec:thrm_instance_ub}}
\begin{proof}(Proof of Theorem \ref{thrm:instance_ub})
The entire analysis is conditional on the fMC p-values $\{P_i^{\text{fMC}}\} = \{p_i\}$.
Without loss of generality assume $p_1\leq p_2\leq \cdots \leq p_m$.
Let $T$ be the total number of rounds, which is finite since at most $mn$ MC samples will be computed. For any round $t$, let ``$(t)$'' represent the corresponding values before the MC sampling of the round. Note that ``$(T+1)$'' represent the values at termination. The quantities useful to the analysis include
\begin{enumerate}
    \item $N_i(t)$: number of MC samples for arm $i$.
    \item $p^{\text{lb}}_i(t), p^{\text{ub}}_i(t)$: lower and upper CBs for arm $i$.
    \item Empirical mean $\hat{p}_i(t) = \frac{1}{N_i(t)} \left(1 \vee \sum_{j=1}^{N_i(t)} B_{i,j} \right)$.
    \item $\mathcal{C}_\text{g}(t)$, $\mathcal{C}_\text{l}(t)$, $\mathcal{U}(t)$: hypothesis sets as defined in \eqref{eq:h_set}. 
    \item $\hat{r}(t), \hat{\tau}(t)$: critical rank estimate and the corresponding BH threshold estimate.
\end{enumerate}

Let $\mathcal{E}$ denote the event that all CBs hold. Since the number of CBs is at most $2mL$ and each of them holds with probability at least $1-\frac{\delta}{2mL}$ conditional on the fMC p-values, by union bound,
\begin{align*}
    \P\left(\mathcal{E} \Big\vert \{P_i^{\text{fMC}}\} = \{p_i\}\right) \geq 1 - \delta.
\end{align*}

Conditional on $\mathcal{E}$, when the algorithm terminates, $\mathcal{U}(T+1) = \emptyset$.
There are $m-r^*$ hypotheses in $\mathcal{C}_{\text{g}}(T+1)$ and $r^*$ hypotheses in $\mathcal{C}_{\text{l}} (T+1)$. 
We next upper the number of MC samples for hypotheses in these two sets separately.

{\bf Step 1. Hypotheses in $\mathcal{C}_{\text{g}}(T+1)$.}
On $\mathcal{E}$, there are $m-r^*$ hypotheses in $\mathcal{C}_{\text{g}}(T+1)$. For any $i \in [m-r^*]$, let $g_i$ be the $i$th hypothesis entering $\mathcal{C}_{\text{g}}$. For two hypotheses entering $\mathcal{C}_{\text{g}}$ in the same round, the one is considered entering earlier if it has a larger upper CB $p^{\text{ub}}$ before the MC sampling in the entering round.

Consider any $g_i$ that enters after MC sampling in round $t_i$ and  
let $g_j$ be the first hypothesis entering $\mathcal{C}_{\text{g}}$ in the same round. 
Here, we note that $t_i = t_j$ and the number of MC samples $N_{g_i}(T+1) = N_{g_j}(T+1)$. 
In addition, 
\begin{align}
    N_{g_j}(T+1) = N_{g_j}(t_j+1) \leq (1+\gamma) N_{g_j}(t_j),
\end{align}
since the batch sizes is a geometric sequence with ratio $\gamma$. 
Now we focus on $N_{g_j}(t_j)$.

Since $g_j$ is sampled in round $t_j$, we have that $g_j \notin \mathcal{C}_\text{g}(t_j)$.
This indicates that in round $t_j$, the lower CB of $g_j$ should be no greater than the estimated threshold $\hat{\tau}(t_j)$ before MC sampling;
otherwise $g_j$ would have entered $\mathcal{C}_\text{g}$ before round $t_j$. Hence,
\begin{align} \label{eq:thrm_bh_pf1}
    p^{\text{lb}}_{g_j} (t_j) \leq \hat{\tau}(t_j).
\end{align}
Also, being the first to enter $\mathcal{C}_\text{g}$ in round $t_j$, its upper CB is the largest among all elements in $\mathcal{U}(t_j)$, i.e., 
\begin{align} \label{eq:thrm_bh_pf2}
    p^{\text{ub}}_{g_j} (t_j) = \max_{k \in \mathcal{U}(t_j)} p^{\text{ub}}_{k} (t_j).
\end{align}
Subtracting \eqref{eq:thrm_bh_pf1} from \eqref{eq:thrm_bh_pf2} to have the width of the confidence interval 
\begin{equation} \label{eq:thrm_bh_pf3}
    \begin{split}
        p^{\text{ub}}_{g_j}(t_j) - p^{\text{lb}}_{g_j} (t_j) & \geq \max_{k \in \mathcal{U}(t_j)} p^{\text{ub}}_{k} (t_j) - \hat{\tau}(t_j) \\
        & \geq \max_{k \in \mathcal{U}(t_j)} p_{k} - \hat{\tau}(t_j),
    \end{split}
\end{equation}
where the last inequality is conditional on $\mathcal{E}$.
Since $\vert \mathcal{C}_\text{g}(t_j) \vert = j-1$, we have that $\max_{k \in \mathcal{U}(t_j)} p_{k} \geq p_{m-j+1}$.
Therefore \eqref{eq:thrm_bh_pf3} can be further written as 
\begin{align} \label{eq:thrm_bh_pf4}
    p^{\text{ub}}_{g_j}(t_j) - p^{\text{lb}}_{g_j} (t_j) \geq p_{m-j+1} - \hat{\tau}(t_j) = \Delta_{m-j+1}.
\end{align}

Since the CBs satisfy \eqref{eq:CB_form}, equations \eqref{eq:thrm_bh_pf1} and \eqref{eq:thrm_bh_pf4} can be rewritten as 
\begin{equation}
    \begin{split}
    & \hat{p}_{g_j} (t_j) - \sqrt{\frac{c\left(\frac{\delta}{2 mL}\right) \hat{p}_{g_j} (t_j) }{T_{g_j}(t_j)}} \leq \hat{\tau}(t_j),\\
    & 2 \sqrt{\frac{c\left(\frac{\delta}{2 mL}\right)\hat{p}_{g_j} (t_j) }{T_{g_j}(t_j)}} \geq \Delta_{m-j+1}.
    \end{split}
\end{equation}

Note that $\hat{\tau}(t_j) = \frac{m-j+1}{m}\alpha$. By Lemma \ref{lm:anci_bh},
\begin{align}\label{eq:thrm_bh_pf5}
    N_{g_j}(t_j) & \leq \frac{4 c\left(\frac{\delta}{2mL}\right) \left(\frac{m-j+1}{m} \alpha + \frac{\Delta_{m-j+1}}{2} \right) }{\Delta_{m-j+1}^2} \\
    & \leq \frac{4 c\left(\frac{\delta}{2mL}\right) p_{m-j+1} }{\Delta_{m-j+1}^2}.
\end{align}
Since $i\geq j$, we have that $m-j+1 \geq m-i+1$. Therefore. 
\begin{equation}\label{eq:thrm_bh_pf6}
    \begin{split}
    \E[N_{g_i}(T+1) \vert \mathcal{E} ] &\leq (1+\gamma) \E[N_{g_i}(t_i) \vert \mathcal{E} ]\\
    & \leq (1+\gamma) \frac{4 c\left(\frac{\delta}{2mL}\right) p_{m-j+1} }{\Delta_{m-j+1}^2} \\
    & \leq \max_{k\geq m-i+1} \frac{4(1+\gamma) c\left(\frac{\delta}{2mL}\right) p_k }{\Delta_{k}^2}.
    \end{split}
\end{equation}

{\bf Step 2. Hypotheses in $\mathcal{C}_{\text{l}}(T+1)$.}
On $\mathcal{E}$, $\mathcal{C}_\text{l}(T+1) = \mathcal{R}^{\mathrm{fMC}}$ and $\hat{\tau}(T+1) = \tau^*$. 
Consider any hypothesis $i \in \mathcal{C}_\text{l}(T+1)$ whose fMC p-value is $p_i \leq \tau^*$.
It will be sampled until its upper CB is no greater than $\tau^*$. Let its last sample round be $t_i$.
Then,
\begin{align}\label{eq:thrm_bh_pf7}
     p^{\text{ub}}_{g_i}(t_i) > \tau^*,~~~~ p^{\text{ub}}_{g_i}(t_i+1)\leq \tau^*, ~~~~ p^{\text{lb}}_{g_i}(t_i) \leq p_i.
\end{align}
Subtracting the third term from the first term yields 
\begin{align}\label{eq:thrm_bh_pf8}
    p^{\text{ub}}_{g_i}(t_i) - p^{\text{lb}}_{g_i}(t_i) > \Delta_i.
\end{align}

Since the CBs satisfy \eqref{eq:CB_form}, the second term in \eqref{eq:thrm_bh_pf7} along with \eqref{eq:thrm_bh_pf8} can be rewritten as 
\begin{equation}\label{eq:thrm_bh_pf9}
    \begin{split}   
    & \hat{p}_{i} (t_i+1) + \sqrt{\frac{c\left(\frac{\delta}{2 mL}\right) \hat{p}_{i} (t_i+1) }{N_{i}(t_i+1)}} \leq \tau^*,\\
    & 2 \sqrt{\frac{c\left(\frac{\delta}{2 mL}\right)\hat{p}_i (t_i) }{N_i(t_i)}} > \Delta_i.
    \end{split}
\end{equation}
Note that $N_{i}(t_i+1) \leq (1+\gamma) N_{i}(t_i)$ and $\hat{p}_{i} (t_i+1)\geq \frac{1}{1+\gamma} \hat{p}_{i} (t_i)$,
\eqref{eq:thrm_bh_pf9} can be further written as 
\begin{align}\label{eq:thrm_bh_pf10}
    \begin{split}
    & \hat{p}_{i} (t_i) + \sqrt{\frac{c\left(\frac{\delta}{2 mL}\right) \hat{p}_{i} (t_i) }{N_{i}(t_i)}} \leq (1+\gamma) \tau^* \\
    & 2 \sqrt{\frac{c\left(\frac{\delta}{2 mL}\right)\hat{p}_i (t_i) }{N_i(t_i)}} > \Delta_i.
    \end{split}
\end{align}

Furthermore, 
\begin{align}
    N_{i}(t_i) \leq \frac{4(1+\gamma) c\left(\frac{\delta}{2 mL}\right)  \tau^*}{\Delta_i^2}.
\end{align}
and the number of MC samples for hypothesis $i$
\begin{equation}
    \begin{split}
    \E[N_{i}(T+1) \vert \mathcal{E} ] &\leq (1+\gamma) \E[N_{i}(t_i) \vert \mathcal{E} ]\\
    & \leq \frac{4(1+\gamma)^2 c\left(\frac{\delta}{2 mL}\right)  \tau^*}{\Delta_i^2}.
    \end{split}
\end{equation}

{\bf Step 3. Combine the result.}
Finally, noting that a hypothesis can be at most sampled $n$ times, the total expected MC samples 
\begin{align}
    \E[N] & \leq \E\left[ \sum_{i=1}^m N_i(T+1) \Big\vert \mathcal{E}\right] + \delta mn \\
    & \leq \sum_{i=1}^{r^*} n \wedge \left( \frac{4(1+\gamma)^2 c\left(\frac{\delta}{2 mL}\right)  \tau^*}{\Delta_i^2}\right) \\
    & \sum_{i=r^*+1}^m n \wedge \left( \max_{k\geq i} \frac{4(1+\gamma) c\left(\frac{\delta}{2mL}\right) p_k }{\Delta_{k}^2} \right)  + \delta mn.
\end{align}
\end{proof}

\section{Proof of Proposition \ref{crly:ub}\label{suppsec:crly_ub}}
\begin{proof}(Proof of Proposition \ref{crly:ub})
First let us consider the case where $f(p)$ is continuous and monotonically decreasing. The case where $f(p)=1$ is easy and is dealt with at the end.

{\bf Step 0. Notations.}
Since this proof is an asymptotic analysis, we use subscript ``${n,m}$'' to denote the quantities for the fMC p-values with $n$ MC samples and $m$ hypotheses. We are interested in the regime where $m\rightarrow \infty$ while $n = \Omega(m)$.

For an instance with $m$ hypotheses and $n$ MC samples for each hypothesis, let $\tilde{\tau}_{n,m}$ be the BH threshold and $\tilde{F}_{n,m}$ be the empirical distribution of the fMC p-values $\tilde{F}_{n,m}(x) = \frac{1}{m}\sum_{i=1}^m \ind\{P_i^{\text{fMC}}\leq x\}$. Also let $\tilde{f}_{n,m}$ be the probability mass function $\tilde{f}_{n,m}(x) = \frac{1}{m}\sum_{i=1}^m \ind\{P_i^{\text{fMC}}= x\}$.

For the distribution of the ideal p-values $F$, define $g(x) = x - F(x) \alpha$ and let $\tau^* = \sup_{[0,1]}\{\tau:g(\tau) \leq 0\}$. 
$\tau^*$ is actually the BH threshold in the limiting case, as will be shown in Step 2 below. There are a few properties we would like to point out. By definition $g(\tau^*) = 0$. As a result, $F(\tau^*) = \frac{\tau^*}{\alpha}$. 
Since $f(p)$ is monotonically decreasing, $f(\tau^*) < \frac{F(\tau^*)}{\tau^*} = \frac{1}{\alpha}$. Furthermore, $g'(\tau^*) = 1 - f(\tau^*) \alpha > 0$.

{\bf Step 1. $\tilde{F}_{n,m}$ converges uniformly to $F$.}
Let $F_n$ be the distribution of the fMC p-values with $n$ MC samples. Then $F_n$ converges uniformly to $F$. Furthermore, by Glivenko-Cantelli theorem $\tilde{F}_{n,m}$ converges uniformly to $F_n$. Therefore, $\tilde{F}_{n,m}$ converges uniformly to $F$.

{\bf Step 2. $\tilde{\tau}_{n,m}$ converges in probability to $\tau^*$.} 
For an instance with $m$ hypotheses and $n$ MC samples for each hypothesis, let $\tilde{g}_{n,m}(x) = x - \tilde{F}_{n,m}(x) \alpha$. Then $\tilde{\tau}_{n,m} = \sup_{[0,1]}\{\tau:\tilde{g}_{n,m}(\tau) \leq 0\}$. Since $\tilde{F}_{n,m}$ converges uniformly to $F$, $\tilde{g}_{n,m}$ converges uniformly to $g$. 
Since $g'(\tau^*) >0$ and is continuous at $\tau^*$, $\exists \epsilon_0>0$ such that $g(x)$ is monotonically increasing on $[\tau^*-\epsilon_0, \tau^*+\epsilon_0]$. Since $\tilde{g}_{n,m}$ converges uniformly to $g$ on this interval, for any $0<\epsilon'<\epsilon$, $\P(\vert \tilde{\tau}_{n,m} - \tau^* \vert > \epsilon') \rightarrow 0$. Thus, $\tilde{\tau}_{n,m} \overset{p}{\rightarrow} \tau^*$.

{\bf Step 3. Upper bound $\E  [N]$.} 
Let $\delta = \frac{1}{mn}$ and let $\tilde{c}$ denote any log factor (in both $m$ and $n$) in general.
Then for the fMC p-values with $n$ MC samples and $m$ hypotheses, by Theorem \ref{thrm:recovery}, and omitting additive constants, 
\begin{equation}\label{eq:ub_pf1}
\begin{split}
    & \E  [N] \leq  \tilde{c} \E \left[ \sum_{i=1}^{r^*} n \wedge \frac{\tilde{\tau}_{n,m}}{\Delta_{(i)}^2}
     + \sum_{i=r^*+1}^m n \wedge \max_{k\geq i} \frac{ P_{(k)}^{\text{fMC}} }{\Delta_{(k)}^2} \right] \\
    & \leq \tilde{c} \E \left[ \sum_{i=1}^{r^*} n \wedge \frac{1}{\Delta_{(i)}^2}
     + \sum_{i=r^*+1}^m n \wedge \max_{k\geq i} \frac{1}{\Delta_{(k)}^2} \right].
\end{split}
\end{equation}
Notice that $\tilde{F}_{n,m}(P_{(k)}^{\text{fMC}}) \geq \frac{k}{m}$ where the inequality is because there might be several hypotheses with the same value.  Therefore for any $P_{(k)}^{\text{fMC}}>\tilde{\tau}_{n,m}$, 
\begin{align*}
    & \frac{1}{\Delta_{(k)}^2} = \frac{1}{\left(P_{(k)}^{\text{fMC}} - \frac{k}{m}\alpha\right)^2} \\
    & \leq \frac{1}{\left(P_{(k)}^{\text{fMC}} - \tilde{F}_{n,m}(P_{(k)}^{\text{fMC}})\alpha\right)^2}  = \frac{1}{\tilde{g}_{n,m}(P_{(k)}^{\text{fMC}})^2}.
\end{align*}
Hence, summing over all possible values of the empirical distribution of the fMC p-values, i.e., $P^{\text{fMC}} = \frac{1}{n+1}, \frac{2}{n+1}, \cdots, 1$ (note the definition of the fMC p-values in \eqref{eq:MC_p_val}), to further write \eqref{eq:ub_pf1} as  
\begin{equation}\label{eq:ub_pf2}
\begin{split}
    & \E  [N] \leq \\
    & \tilde{c} m \E\left[ \sum_{i=1}^{\lfloor (n+1) \tilde{\tau}_{n,m}\rfloor} \left(n \wedge  \frac{1}{\left(\frac{i}{n+1} - \tilde{\tau}_{n,m}\right)^2} \right)\tilde{f}_{n,m}\left(\frac{i}{n+1}\right) \right.\\
    &  + \left.\sum_{i=\lceil (n+1) \tilde{\tau}_{n,m}\rceil}^{n+1} \left( n \wedge  \max_{k\geq i} \frac{1}{\tilde{g}_{n,m}(\frac{k}{n+1})^2} \right) \tilde{f}_{n,m}\left(\frac{i}{n+1}\right)\right].
\end{split}
\end{equation}

Since $f(x)$ is continuous, $g'(x)$ is also continuous. Recall that $g'(\tau^*)>0$.
Hence, $\exists \epsilon, c_0 > 0$ such that $\forall x \in [\tau^* - \epsilon, 1]$, $g'(x)>c_0$.
Recall that $\tilde{g}_{m,n}$ converges uniformly to $g$ and $\tilde{\tau}_{n,m} \overset{p}{\rightarrow}\tau^*$.
Note that by definition $\tilde{g}_{n,m}(\tilde{\tau}_{n,m})=0$. Therefore, $\exists c_1 > 0$ such that for large enough $n,m$, for any $k \geq \lceil (n+1) \tilde{\tau}_{n,m}\rceil$, 
\begin{align}
    & \tilde{g}_{n,m}\left(\frac{k}{n+1}\right) = \tilde{g}_{n,m}\left(\frac{k}{n+1}\right) - \tilde{g}_{n,m}(\tilde{\tau}_{n,m}) \\
    & \geq c_1 \left(\frac{k}{n+1} - \tilde{\tau}_{n,m}\right).
\end{align}

Hence, \eqref{eq:ub_pf2} can be further rewritten as 
\begin{align*}
    & \E  [N] \leq \\
    & \tilde{c} m \E\left[ \sum_{i=1}^{\lfloor (n+1) \tilde{\tau}_{n,m}\rfloor} \left(n \wedge  \frac{1}{\left(\frac{i}{n+1} - \tilde{\tau}_{n,m}\right)^2} \right)\tilde{f}_{n,m}\left(\frac{i}{n+1}\right) \right.\\
    &  + \left.\sum_{i=\lceil (n+1) \tilde{\tau}_{n,m}\rceil}^{n+1} \left( n \wedge  \frac{1}{c_1^2\left(\frac{i}{n+1} - \tilde{\tau}_{n,m}\right)^2} \right) \tilde{f}_{n,m}\left(\frac{i}{n+1}\right)\right]\\
    & \leq \frac{\tilde{c}}{c_1^2 } m \E\left[ \sum_{i=1}^{n+1} n \wedge \frac{1}{(\frac{i}{n} - \tilde{\tau}_{n,m})^2} \tilde{f}_{n,m}\left(\frac{i}{n+1}\right)\right]\\
    & = \frac{\tilde{c}}{c_1^2 }  m \E\left[ n \wedge \frac{1}{(P_i^{\text{fMC}}- \tilde{\tau}_{n,m})^2}\right].
\end{align*}

Since $F_{n}$ converges uniformly to $F$ and $\tilde{\tau}_{n,m} \overset{p}{\rightarrow} \tau^*$, by Slutsky's theorem and the continuous mapping theorem, the RHS will converge to 
\begin{align}
    \frac{\tilde{c}}{c_1^2 } m \E\left[ n \wedge \frac{1}{(P_i^{\infty} - \tau^*)^2}\right].
\end{align}
Last we evaluation the expectation: 
\begin{equation*}
    \begin{split}
    & \E\left[ n \wedge \frac{1}{(P_i^{\infty} - \tau^*)^2}\right] = \int_0^{\tau^* - \frac{1}{\sqrt{n}}} \frac{1}{(p-\tau^*)^2}dF(p) \\
    & + \int_{\tau^* - \frac{1}{\sqrt{n}}}^{\tau^* + \frac{1}{\sqrt{n}}} n dF(p) + \int_{\tau^* + \frac{1}{\sqrt{n}}}^1 \frac{1}{(p-\tau^*)^2}dF(p).
    \end{split}
\end{equation*}
By noting that $f(\tau^*)< \frac{1}{\alpha}$ and $f(p)$ is monotonically decreasing it is clear that all three terms are $\tilde{O}(\sqrt{n})$, which concludes the proof of this case.  

When $f(p)=1$, the limiting BH threshold $\tau^* = 0$. Furthermore, $g(x)=(1-\alpha)x$ and $g'(x) = 1-\alpha >0$. Therefore, $g(\frac{k}{n+1})\geq (1-\alpha)(\frac{k}{n+1} - \tilde{\tau}_{n,m})$. Then, similarly we have the total number of MC samples 
\begin{align}
    \E[N] \leq \frac{\tilde{c}}{(1-\alpha)^2 }  m \E\left[ n \wedge \frac{1}{(P_i^{\text{fMC}}- \tilde{\tau}_{n,m})^2}\right],
\end{align}
which converges to 
\begin{align}
    \frac{\tilde{c}}{(1-\alpha)^2 }  m \E\left[ n \wedge \frac{1}{(P_i^{\infty})^2}\right]
\end{align}
that is $\tilde{O}(\sqrt{n}m)$.
\end{proof}

\section{Proof of Theorem \ref{thrm:lb}\label{suppsec:thrm_lb}}

\begin{proof}(Proof of Theorem \ref{thrm:lb})
Let $F_n$ be the distribution of the fMC p-values with $n$ MC samples.
By Lemma \ref{lm:instance_lb}, conditional on the fMC p-values $\{P_i^{\text{fMC}}\} = \{p_i\}$, $\exists \delta_0>0$, $c_0>0$, $c_1>0$, s.t. $\forall \delta<\delta_0$, a $\delta$-correct algorithm satisfies 
\begin{align}
    \E\left[N \Big\vert \{P_i^{\text{fMC}}\} = \{p_i\}\right] \geq  c_0 n \sum_{i=1}^m  \ind\{\tau^* < p_i \leq \tau^*+\frac{c_1}{\sqrt{n}}\}.
\end{align}
Taking expectation with respect to the fMC p-values to have 
\begin{align}\label{eq:lb_pf_1}
    \E\left[N \right] \geq c_0 n m \P\left[ \tau^* < P_i^{\text{fMC}} \leq \tau^*+\frac{c_1}{\sqrt{n}}\right].
\end{align}
Since the null fMC p-values follow a uniform distribution, 
\begin{align}\label{eq:lb_pf_2}
    \E\left[N \right] \geq c_0 \pi_0 n m\frac{c_1}{\sqrt{n}} = c_0 c_1 \pi_0 \sqrt{n}m,
\end{align}
which completes the proof. 

\end{proof}

\section{Auxiliary Lemmas\label{sec:lm}}
\begin{lemma}\label{lm:anci_bh}
    For $c>0$, $\hat{p}>0$, $\Delta>0$, $\tau>0$, if 
    \begin{align} \label{eq:lm_anci_bh_2}
        \hat{p} - \sqrt{\frac{c \hat{p}}{n}} \leq \tau,
         ~~~~ 2 \sqrt{\frac{c \hat{p}}{n}} \geq \Delta,
    \end{align}
    then 
    \begin{align}
        n \leq \frac{4c (\tau + \frac{\Delta}{2})}{\Delta^2}.
    \end{align}
\end{lemma}
\begin{proof}(Proof of Lemma \ref{lm:anci_bh})
    Rearranging the first inequality in \eqref{eq:lm_anci_bh_2} and taking square of both sides to have
    \begin{align*}
        \hat{p}^2 - 2\tau \hat{p} + \tau^2  \leq \frac{c \hat{p}}{n}.
    \end{align*}
    This further gives that 
    \begin{align*}
        \hat{p} \leq \tau + \frac{c}{2n} + \sqrt{\frac{c}{n}\tau + \frac{c^2}{4n^2}}.
    \end{align*}
    Combining the above with the second inequality in \eqref{eq:lm_anci_bh_2} to have 
    \begin{align*}
        \frac{\Delta^2}{4c} n \leq \hat{p} \leq  \tau + \frac{c}{2n} + \sqrt{\frac{c}{n}\tau + \frac{c^2}{4n^2}},
    \end{align*}
    which can be rearranged as 
    \begin{align*}
        \frac{\Delta^2}{4c} n - \tau - \frac{c}{2n} \leq \sqrt{\frac{c}{n}\tau + \frac{c^2}{4n^2}}.
    \end{align*}
    Taking square of both sides and cancel the repeated terms to have 
    \begin{align*}
        \left( \frac{\Delta^2}{4c} n \right)^2 - \frac{\Delta^2\tau }{2c}n + \tau ^2 - \frac{\Delta^2}{4} \leq 0,
    \end{align*}
    which is equivalent to 
    \begin{align*}
        \left( \frac{\Delta^2}{4c} n - \tau \right)^2 \leq \frac{\Delta^2}{4}.
    \end{align*}
    Taking square root of both sides and we completed the proof. 
\end{proof}

\begin{lemma}\label{lm:instance_lb}
Given the fMC p-values $\{P^{\text{fMC}}_i\} = \{p_i\}$ with BH threshold $\tau^*$, 
$\exists \delta_0 \in (0, 0.5)$, $c_0>0$, $c_1>0$, s.t. $\forall \delta<\delta_0$, a $\delta$-correct algorithm satisfies
\begin{align*}
    \E\left[N \Big\vert \{P_i^{\text{fMC}}\} = \{p_i\}\right] \geq  c_0 n \sum_{i=1}^m  \ind\{\tau^* < p_i \leq \tau^*+\frac{c_1}{\sqrt{n}}\}.
\end{align*}
\end{lemma}
\begin{proof} (Proof of Lemma \ref{lm:instance_lb})
Consider any $\delta$-correct algorithm and let us denote the true (unknown) fMC p-values by $\{q_i\}$. 
For any null hypothesis $l$ with fMC p-value $\tau^*<p_l\leq \tau^* + \frac{c_1}{\sqrt{n}}$, consider the following settings:
\begin{align}
    & H_0: q_i = p_i,~~~~\text{for}~i \in [m], \\
    & H_l: q_l = \tau^*,~~~~q_i=p_i,~\text{for}~i \neq l.
\end{align}
The $\delta$-correct algorithm should accept the $l$th null hypothesis under $H_0$ and reject it under $H_l$, both with probability at least $1-\delta$. 
For $x \in \{0, l\}$, we use $\E_x$ and $\P_x$ to denote the expectation and probability, respectively, conditional on the fMC p-values $\{P^{\text{fMC}}_i\} = \{q_i\}$, under the algorithm being considered and under setting $H_x$. Let $N_l$ be the total number of MC samples computed for null hypothesis $l$. 
In order to show Lemma \ref{lm:instance_lb}, it suffices to show that $\E_0[N_l]\geq c_0 n$. 
We prove by contradiction that if $\E_0[N_l] < c_0 n$ and if the algorithm is correct under $H_0$ with probability at least $0.5$, the probability that it makes a mistake under $H_l$ is bounded away from 0.

{\bf Notations.}
Let $S_{l,t}$ to be the number of ones when $t$ MC samples are collected for the $l$th null hypothesis.
We also let $S_l$ be the number of ones when all $N_l$ MC samples are collected. 
Let $k_0=(n+1)p_l-1$ and $k_l=(n+1)\tau^*-1$. 
Given $N_l$, $S_l$ follows hypergeometric distribution with parameters $(N_l, k_0, n)$ and $(N_l, k_l, n)$ under $H_0$ and $H_l$, respectively. 
Let $\Delta_k = k_0 - k_l$. We note that  
\begin{align}
    \Delta_k = (n+1)(p_l-\tau^*) \in (0, \frac{c_1 (n+1)}{\sqrt{n}}].
\end{align}

{\bf Define key events.}
Let $c_0=1/8$ and define the event 
\begin{align}
    \mathcal{A}_l = \{N_l \leq 0.5n\}. 
\end{align}
Then by Markov's inequality, $\P_0(\mathcal{A}_l) \geq \frac{3}{4}$.

Let $\mathcal{B}_l$ be the event that the $l$th null hypothesis is accepted. Then $\P_0(\mathcal{B}_l) \geq 1-\delta > 1/2$. 

Let $\mathcal{C}_l$ be the event defined by 
\begin{align}
    \mathcal{C}_l = \left\{ \max_{1\leq t \leq 0.5 n} \vert S_{l,t} - t k_0/n \vert < 2 \sqrt{n}\right\}.
\end{align}
By Lemma \ref{lm:maximal} $\P_0(\mathcal{C}_l) \geq 7/8$.

Finally, define the event $\mathcal{S}_l$ by $\mathcal{S}_l = \mathcal{A}_l \cap \mathcal{B}_l \cap \mathcal{C}_l$. Then $\P_0(\mathcal{S}_l) > 1/8$.

{\bf Lower bound the likelihood ratio. } We let $W$ be the history of the process (the sequence of null hypotheses chosen to sample at each round, and the sequence of observed MC samples) until the algorithm terminates.
We define the likelihood function $L_l$ by letting
\begin{align}
    L_l(w) = \P_l(W=w),
\end{align}
for every possible history $w$. Note that this function can be used to define a random variable $L_l(W)$.

Given the history up to round $t-1$, the null hypotheses to sample at round $t$ has the same probability distribution under either setting $H_0$ and $H_l$; similarly, the MC sample at round $t$ has the same probability setting, under either hypothesis, except for the $l$th null hypothesis. For this reason, the likelihood ratio
\begin{equation}\label{eq:instance_lb_pd_1}
\begin{split}
    & \frac{L_l(W)}{L_0(W)} = \frac{\binom{k_l}{S_l} \binom{n-k_l}{N_l-S_l}}{\binom{k_0}{S_l} \binom{n-k_0}{N_l-S_l}} \\
    & = \prod_{r=0}^{S_l-1} \frac{k_l -r}{k_0 -r} \prod_{r=0}^{N_l-S_l-1} \frac{n-k_l-r}{n-k_0-r} \\
    & = \prod_{r=0}^{S_l-1} \left( 1 - \frac{\Delta_k}{k_0 -r}\right) \prod_{r=0}^{N_l-S_l-1} \left( 1+ \frac{\Delta_k}{n-k_0-r}\right) \\
\end{split}
\end{equation}
Next we show that on the event $\mathcal{S}_l$, the likelihood ratio is bounded away from 0. 

If $S_l \leq 100\sqrt{n}$, then the likelihood ratio 
\begin{align}
    & \frac{L_l(W)}{L_0(W)} \geq \left( 1 - \frac{\Delta_k}{k_0 - S_l}\right)^{S_l} \geq \left( 1 - \frac{c_2}{\sqrt{n}} \right)^{100\sqrt{n}} > c_3,
\end{align}
for some constants $c_2>0$, $c_3>0$.

If $S_l > 100\sqrt{n}$, further write \eqref{eq:instance_lb_pd_1} as 
\begin{equation}\label{eq:instance_lb_pd_2}
\begin{split}
    & \frac{L_l(W)}{L_0(W)} 
    = \prod_{r=0}^{S_l-1} \left\{\left[ 1 - \left(\frac{\Delta_k}{k_0 -r}\right)^2\right] \left( 1 + \frac{\Delta_k}{k_0 -r}\right)^{-1}\right\} \\
    & \prod_{r=0}^{N_l-S_l-1} \left( 1+ \frac{\Delta_k}{n-k_0-r}\right).
\end{split}
\end{equation}

Since $S_l > 100\sqrt{n}$, on $\mathcal{C}_l$, $\frac{N_l-S_l}{S_l} > 1$. 
Note that if $a\geq 1$, then the mapping $x \mapsto (1+x)^a$ is convex for $x>-1$. Thus, $(1+x)^a \geq 1+ax$, which implies that for any $0 \leq r \leq k_0$,
\begin{align}\label{eq:instance_lb_pd_3}
    \left( 1 + \frac{\Delta_k}{\frac{N_l-S_l}{S_l}(k_0-r)}\right)^{\frac{N_l-S_l}{S_l}} \overset{\mathcal{C}_l}{\geq} \left( 1 + \frac{\Delta_k}{k_0 -r}\right).
\end{align}
Then, \eqref{eq:instance_lb_pd_2} can be further written as 
\begin{equation}\label{eq:instance_lb_pd_4}
\begin{split}
    & \frac{L_l(W)}{L_0(W)} \overset{\eqref{eq:instance_lb_pd_3}}{\geq} \prod_{r=0}^{S_l-1} \left[ 1 - \left(\frac{\Delta_k}{k_0 -r}\right)^2\right] \\
    &  \prod_{r=0}^{S_l-1} \left( 1 + \frac{\Delta_k}{\frac{N_l-S_l}{S_l}(k_0-r)}\right)^{ - \frac{N_l-S_l}{S_l}}\\
    &  \prod_{r=0}^{N_l-S_l-1} \left( 1+ \frac{\Delta_k}{n-k_0-r}\right).
\end{split}
\end{equation}
Note that the 2nd term is no less than 
\begin{align}
    \prod_{r=0}^{N_l-S_l-1} \left( 1 + \frac{\Delta_k}{\frac{N_l-S_l}{S_l}k_0-r}\right)^{ - 1}.
\end{align}
Eq. \eqref{eq:instance_lb_pd_4} can be further written as 
\begin{equation}\label{eq:instance_lb_pd_5}
\begin{split}
    & \frac{L_l(W)}{L_0(W)} \geq \prod_{r=0}^{S_l-1} \left[ 1 - \left(\frac{\Delta_k}{k_0 -r}\right)^2\right]\\
    &  \prod_{r=0}^{N_l-S_l-1} \left [ \left( 1 + \frac{\Delta_k}{\frac{N_l-S_l}{S_l}k_0-r}\right)^{ - 1} \left( 1+ \frac{\Delta_k}{n-k_0-r}\right) \right]
\end{split}
\end{equation}
Next we show that both terms in \eqref{eq:instance_lb_pd_5} are bounded away from 0.

{\bf First term in \eqref{eq:instance_lb_pd_5}}
\begin{align}
    & \prod_{r=0}^{S_l-1} \left[ 1 - \left(\frac{\Delta_k}{k_0 -r}\right)^2\right] \geq \left[ 1 - \left(\frac{\Delta_k}{k_0 -S_l}\right)^2\right]^{S_l}\\
    & \overset{\mathcal{A}_l,\mathcal{C}_l}{\geq } \left( 1 - \frac{c_4}{n}\right)^{n} \geq c_5 > 0,
\end{align}
for some constants $c_4>0$, $c_5>0$.

{\bf Second term in \eqref{eq:instance_lb_pd_3}}
\begin{equation}
    \begin{split}
    & \prod_{r=0}^{N_l-S_l-1} \left [ \left( 1 + \frac{\Delta_k}{\frac{N_l-S_l}{S_l}k_0-r}\right)^{ - 1} \left( 1+ \frac{\Delta_k}{n-k_0-r}\right) \right] \\
    & = \prod_{r=0}^{N_l-S_l-1} \left( 1 + \frac{\frac{\Delta_k}{n-k_0 - r} - \frac{\Delta_k}{\frac{N_l-S_l}{S_l}k_0-r}}{1 + \frac{\Delta_k}{\frac{N_l-S_l}{S_l}k_0-r}} \right)\\
    & = \prod_{r=0}^{N_l-S_l-1} \left( 1 + \frac{ \Delta_k  \frac{N_l}{S_l}\left( k_0 - \frac{S_l}{N_l} n \right)}{\left(1 + \frac{\Delta_k}{\frac{N_l-S_l}{S_l}k_0-r}\right) (n-k_0 -r) (\frac{N_l-S_l}{S_l}k_0-r)} \right) \\
    &  \overset{\mathcal{A}_l,\mathcal{C}_l, S_l > 100\sqrt{n}}{\geq} \left( 1 - \frac{c_6}{N_l \sqrt{n}} \right)^{N_l} \geq c_7,
    \end{split}
\end{equation}
for some constants $c_4>0$, $c_5>0$.

Hence $\exists c_8>0$, such that on $\mathcal{S}_l$ the likelihood ratio 
\begin{align}
    \frac{L_l(W)}{L_0(W)} \geq c_8 > 0. 
\end{align}

Therefore, the probability of making an error under $H_l$
\begin{equation}
    \begin{split}
        & \P_l(\text{error}) \geq \P_l(\mathcal{S}_l) = \E_l[\ind\{S_l\}] \\
    & = \E_0\left[\ind\{S_l\} \frac{L_l(W)}{L_0(W)} \right] \geq c_8 \P_0(\mathcal{S}_l) \geq \frac{c_8}{8}.
    \end{split}
\end{equation}
Hence, there does not exist a $\delta$-correct algorithm for any $\delta \leq \frac{c_8}{8}$, completing the proof. 
\end{proof}

\begin{lemma} \label{lm:maximal}
Let $X_1, \cdots, X_n$ be random variables sampled without replacement from the set $\{x_1, \cdots, x_N\}$, where $n \leq N$ and $x_i \in \{0,1\}$. Let $\mu = \frac{1}{N}\sum_{i=1}^N x_i$ and for $k \in [N]$, let $S_k = \sum_{i=1}^k X_i$. Then for any $\theta>0$,
\begin{align}
    \P\left(\max_{1\leq k \leq n} \vert S_k - \mu k \vert \geq \sqrt{n \theta} \right)  \leq \frac{1}{\theta}.
\end{align}
\end{lemma}
This is a direct consequence of Corollary 1.2 in the paper \cite{serfling1974probability}.

\end{document}